\documentclass[11pt]{article}
\usepackage[margin=1in]{geometry}
\usepackage{listings}
\usepackage{amsmath, amssymb, amsthm, amsfonts,cite,alltt,clrscode}
\usepackage[english]{babel}
\usepackage{graphicx}
\usepackage{caption}
\usepackage{accents}
\usepackage[list=true,listformat=simple]{subcaption}
\usepackage{color,colortbl}
\definecolor{almond}{rgb}{0.94, 0.87, 0.8}
\usepackage{hyperref}

\newif\ifabstract
\abstractfalse
\newif\iffull
\ifabstract \fullfalse \else \fulltrue \fi

\ifabstract  
\usepackage{times}

{\makeatletter
 \gdef\section{\@ifnextchar*\section@star\section@normal}
 \gdef\section@normal#1{\refstepcounter{section}%
   \paragraph{\arabic{section}\hbox{~~}#1.}%
   \addcontentsline{toc}{section}{\protect\numberline{\arabic{section}}{#1}}}
 \gdef\section@star*#1{\paragraph{#1.}}}

{\makeatletter
 \gdef\subsection{\@ifnextchar*\subsection@star\subsection@normal}
 \gdef\subsection@normal#1{\refstepcounter{subsection}%
   \paragraph{\thesubsection\hbox{~~}#1.}%
   \addcontentsline{toc}{subsection}{\protect\numberline{\thesubsection}{#1}}}
 \gdef\subsection@star*#1{\paragraph{#1.}}}

\newlength\aboveparagraphskip
\newlength\belowparagraphskip
\makeatletter
\def\paragraph{\@startsection{paragraph}{4}{\z@}{-\aboveparagraphskip}%
                 {\belowparagraphskip}{\normalfont\normalsize\bfseries}}
\makeatother
\fi

\newtheorem{theorem}{Theorem}[section]
\newtheorem{lemma}[theorem]{Lemma}

\newtheorem{corollary}[theorem]{Corollary}

\newcommand{\ul}[1] {\underline{#1}}

{\makeatletter
 \gdef\xxxmark{%
   \expandafter\ifx\csname @mpargs\endcsname\relax 
     \expandafter\ifx\csname @captype\endcsname\relax 
       \marginpar{xxx}
     \else
       xxx 
     \fi
   \else
     xxx 
   \fi}
 \gdef\xxx{\@ifnextchar[\xxx@lab\xxx@nolab}
 \long\gdef\xxx@lab[#1]#2{\textbf{[\xxxmark #2 ---{\sc #1}]}}
 \long\gdef\xxx@nolab#1{\textbf{[\xxxmark #1]}}
}

{\catcode`\^^M=\active \gdef\tablines{\catcode`\^^M=\active \def^^M{\\}}}
\newenvironment{pcode}{\tablines \topsep=0pt \partopsep=0pt \tabbing
   MM\=MM\=MM\=MM\=MM\=MM\=MM\=MM\=MM\=MM\=MM\=MM\=MM\=\kill}
  {\endtabbing\vspace*{-0\baselineskip}}
{\makeatletter \gdef\lasttab{\ifnum \@curtab<\@hightab \>\lasttab\fi}}

\def\compactify{\itemsep=0pt \topsep=0pt \partopsep=0pt \parsep=0pt}
\let\latexusecounter=\usecounter
\newenvironment{itemize*}
  {\begin{itemize}\compactify}
  {\end{itemize}}
\newenvironment{enumerate*}
  {\def\usecounter{\compactify\latexusecounter}
   \begin{enumerate}}
  {\end{enumerate}\let\usecounter=\latexusecounter}
\newenvironment{description*}
  {\begin{description}\compactify}
  {\end{description}}

\def\id#1{\textit{#1}}
\def\proc#1{\textsc{#1}}
\let\epsilon=\varepsilon
\let\keyw=\textbf

\newcommand{\lsuper}{L}
\newcommand{\lrtos}{\lsuper[r\mathbin{:}s]}
\newcommand{\lstot}{\lsuper[s+1\mathbin{:}t]}
\newcommand{\lrtot}{\lsuper[r\mathbin{:}t]}

\begin{document}

\title{Energy-Efficient Algorithms}
\author{%
  Erik D. Demaine%
    \thanks{MIT Computer Science and Artificial Intelligence Laboratory,
      32 Vassar Street, Cambridge, MA 02139, USA,
      \protect\url{{edemaine,jaysonl,geronm,ntyagi}@mit.edu}.
      Supported in part by the MIT Energy Initiative and by
      MADALGO --- Center for Massive Data Algorithmics ---
      a Center of the Danish National Research Foundation.}
\and
  Jayson Lynch\footnotemark[1]
\and
  Geronimo J. Mirano\footnotemark[1]
\and
  Nirvan Tyagi\footnotemark[1]
}
\date{\today}
\maketitle
\begin{abstract}

We initiate the systematic study of the \emph{energy complexity} of
algorithms (in addition to time and space complexity)
based on Landauer's Principle in physics, which gives a lower bound
on the amount of energy a system must dissipate if it destroys information.
We propose energy-aware variations of three standard models of computation: circuit RAM, word RAM, and transdichotomous RAM.
On top of these models, we build familiar high-level primitives such as control
logic, memory allocation, and garbage collection with zero energy complexity
and only constant-factor overheads in space and time complexity,
enabling simple expression of energy-efficient algorithms.
We analyze several classic algorithms in our models and develop low-energy
variations: comparison sort, insertion sort, counting sort, breadth-first
search, Bellman-Ford, Floyd-Warshall, matrix all-pairs shortest paths,
AVL trees, binary heaps, and dynamic arrays.
We explore the time/space/energy trade-off and develop several general
techniques for analyzing algorithms and reducing their energy
complexity.  These results lay a theoretical foundation for a new field
of semi-reversible computing and provide a new framework for the
investigation of algorithms.

\end{abstract}


\thispagestyle{empty}
\setcounter{page}{-1}
\pagebreak

\tableofcontents 
\vfil
\clearpage

\section{Introduction} 

\paragraph{Landauer limit.}

CPU power efficiency (number of computations per kilowatt hour of energy)
has doubled every 1.57 years from 1946 to 2009
\cite{Koomey-Berard-Sanchez-Wong-2011}.
Within the next 15--60 years, however, this trend will hit a fundamental
limit in physics, known as Landauer's Principle \cite{landauer61}.
This principle states that discarding one bit of information
(increasing the entropy of the environment by one bit)
requires $k T \ln 2$ energy, where $k$ is Boltzmann's constant and
$T$ is ambient temperature,
which is about $2.8 \cdot 10^{-21}$ joules or
$7.8 \cdot 10^{-28}$ kilowatt hours at room temperature ($20^\circ$C).
(Even at liquid nitrogen temperatures, this requirement goes down by less
 than a factor of~$5$.)
Physics has proved this principle under a variety of different assumptions
\cite{landauer61,PhysRevA.61.062314,ladyman2007connection,PhysRevE.79.031105,lambson2011exp},
and it has also been observed experimentally
\cite{Landauer-Nature-2012}.

Most CPUs discard many bits of information per clock cycle, as much as
one per gate; for example, an AND gate with output 0 or an OR gate with output 1 ``forgets'' the exact values of its inputs.
To see how this relates to Landauer's principle,
consider the state-of-the-art 15-core
Intel Xeon E7-4890 v2 2.8GHz CPU.  In a 4-processor configuration,
it achieves $1.2 \cdot 10^{12}$ computations per second at 620 watts,%
\footnote{We follow Koomey et al.'s \cite{Koomey-Berard-Sanchez-Wong-2011}
  definitions, using
  \href{http://www.cisco.com/c/dam/en/us/products/collateral/servers-unified-computing/ucs-c460-m4-rack-server/c460m4-specint-rate-base.pdf}{Cisco's measured SPECint\_rate\_base2006 of $2{,}320$}
  to estimate millions of computations per second (MCPS).}
for a ratio of $7.4 \cdot 10^{15}$ computations per kilowatt hour.
At the pessimistic extreme,
if every one of the $4.3 \cdot 10^9$ transistors discards a bit,
then the product $3.2 \cdot 10^{25}$ is only three orders of magnitude greater than
 Landauer limit.  If CPUs continue to double in energy efficiency
every 1.57 years, this gap will close in less than 18 years.
At the more optimistic extreme, if a 64-bit computation discards only 64 bits
(to overwrite one register), the gap will close within 59 years.
The truth is probably somewhere in between these extremes.



\paragraph{Reversible computing.}
The only way to circumvent the Landauer limit is to do logically
\emph{reversible} computations, whose inputs can be reconstructed
from their outputs, using physically \emph{adiabatic} circuits.
According to current knowledge, such computations have no classical
fundamental limitations on energy consumption.
General-purpose CPUs with adiabatic circuits were constructed by Frank and
Knight at MIT \cite{Frank99}. The design of reversible computers is still being
actively studied, with papers on designs for adders \cite{adder}, multipliers \cite{mult}, ALUs \cite{morrison2011design}, clocks \cite{clock}, and processors \cite{BobISA}
being published within the last five years.
AMD's CPUs since Oct.~2012 (Piledriver) use ``resonant clock mesh technology''
(essentially, an adiabatic clock circuit) to reduce overall energy consumption
by 24\% \cite{cyclos}. Thus the ideas from reversible computing are already
creating energy savings today.

But what can be done by reversible computation?
Reversible computation is an old idea, with reversible Turing machines being proved universal by Lecerf in 1963 \cite{lecerf63} and ten years later by Bennett
\cite{bennett73}.  Early complexity results showed that
any computation can be made reversible, but with a quadratic space overhead
\cite{space} or an exponential time overhead \cite{time,time-space00},
in particular models of computation.
More recent results give a trade-off with subquadratic space and subexponential
time \cite{BTV01}.  These general transforms are too expensive; in particular,
in a bounded-space system, consuming extra space to make computations
reversible is just delaying the inevitable destruction of bits.

The relationship between thermodynamics and information theory is described by Zurek \cite{Zurek89}. In a series of papers, Li, Tromp, and Vitanyi discuss irreversible operations as a useful metric for energy dissipation in computers and study the trade-off between time, space, and irreversible operations. An energy cost based on Kolmogorov complexity \cite{Li97anintroduction}, a precise but uncomputable measure of the information content of a string, is introduced in \cite{Li92theoryof} and further explored in \cite{bennett1998information, Vitanyi05, li1996reversibility}. These papers study algorithms for Bennett's pebble game as well as simulating Turing machines; however, they still focus on universal results, eshrew RAM models, and analyze problems more from a complexity than an algorithms perspective.

Irreversibility is just one source of energy consumption in current chips, and
several other models of computation attempt to capture them individually:
switching energy of VLSI circuits \cite{Kissin91}, dynamic and leakage power
loss in CMOS circuits \cite{Korthikanti10,Korthikanti11}, and I/O or memory
access \cite{jain2005}.
Albers \cite{Albers2010} surveys many algorithmic techniques for reducing
energy consumption of current computers, including techniques like sleep states
and power-down mechanisms, dynamic speed scaling, temperature management, and
energy-minimizing scheduling.
Ultimately, however, we believe that irreversibility will become a critical
energy cost shaping the future of computing, and a topic now ripe for
algorithmic analysis.

\paragraph{Our results.}
This paper is the first to perform a thorough algorithmic study of partially
reversible computing, and to analyze realistic time/space/energy trade-offs.
We define the (Landauer/irreversibility) energy cost, and use it to explore
reversible computing in a novel manner.
Although there are many other sources of energy inefficiency in a computer
we believe the Landauer energy cost is a fundamental and useful measure.
A key perspective shift from most of the reversible computing literature
(except \cite{li1996reversibility}) is that we allow algorithms to destroy
bits, and measure the number of destroyed bits as the energy cost.
This approach enables the unification of classic time/space measures
with a new energy measure.
In particular, it enables us to require algorithms to properly clean up all
additional space by the end of their execution, and data structures to be
properly charged for their total space allocation.

We introduce three basic models for analyzing the energy cost of word-level
operations, similar to the standard word models used in most algorithms today:
the word RAM, the more general transdichotomous RAM, and the realistically
grounded circuit RAM.  Our models allow arbitrary computation to be performed,
but define a spectrum of ``irreversibility'', from reversible (free)
computation to completely destructive (expensive) computation.
On top of these basic models (akin to assembly language), we build a high-level
pseudocode for easy algorithm specification, by showing how to implement many
of the familiar high-level programming structures, as well as some new
structures,
with zero energy overhead and only constant-factor overheads in time and space:
\begin{enumerate*}
\item \textbf{Control logic:} If/then/else, for/while loops, jumps, function calls, stack-allocated variables.
\item \textbf{Memory allocation:} Dynamic allocation and deallocation of
  fixed-size or variable-size blocks, in particular implementing
  pointer-machine algorithms.
\item \textbf{Garbage collection:} Reference-counting and mark-and-sweep
  algorithms for finding no-longer-used memory blocks for automatic deallocation.
\item \textbf{Logging and unrolling:} Specific to energy-efficient computation,
  we describe a new programming-language feature that makes it easy to turn
  energy into space overhead, and later remove that space overhead by playing it backwards.
\end{enumerate*}

These models open up an entire research field, which we call
\emph{energy-efficient algorithms}, to find the minimum energy required
to solve a desired computational problem within given time and space bounds.
We launch this field with several initial results about classic algorithmic
problems, first analyzing the energy cost of existing algorithms, and
then modifying or designing new algorithms to reduce the energy cost
without significantly increasing the time and space costs.
Table~\ref{algorithms table} summarizes these results.

Although there are many practical papers about minimizing energy in
computation (favoring instructions that use somewhat less energy than others),
the algorithms community has not made it a standard measure to complement time
and space because, without the idea of reversibility,
energy is simply within a constant-factor of time.
By contrast, in our model, the energy cost can be anywhere between $0$
(for reversible computation) and $t \cdot w$ where $t$ is the running time
(number of word operations) and $w$ is the number of bits in a word.

\def\rowheading#1{\multicolumn{5}{l}{\cellcolor{black}\textcolor{white}{\textbf{#1}}}}
\def\columnheading#1{\multicolumn{1}{c}{\cellcolor{blue}\textcolor{white}{\textbf{#1}}}}
\def\unheading#1{\rm #1}

\begin{table}[t]
\centering
\tabcolsep=0.8\tabcolsep
\begin{tabular}{l|l|l|l|c}
\columnheading{Primitive} & \columnheading{Time} & \columnheading{Space in Log }& \columnheading{Energy}& \columnheading{Thm.}  \\ \hline
\columnheading{} & \columnheading{\unheading (ops)}  & \columnheading{\unheading (bits)} & \columnheading{\unheading (bits)} & \columnheading{} \\ \hline
\rowheading{Control Logic} \\ \hline
Paired Jump & $\Theta(1)$ & $1$ & 0 & \ref{thm:jump}\\ \hline
Variable Jump & $\Theta(1)$ & $1+w$ & 0 & \ref{thm:jump}\\ \hline
Protected If & $\Theta(1)$ & 0 & 0 & \ref{thm:if}\\ \hline
General If & $\Theta(1)$ & $1$ & 0 & \ref{thm:if}\\ \hline
Simple For loop & $\Theta(l)$ & 0 & 0 & \ref{thm:simplefor}\\ \hline
Protected For loop & $\Theta(l)$ & 0 & 0 & \ref{thm:protectedfor}\\ \hline
General For loop & $\Theta(l)$ & $\lg{l}$ & 0 & \ref{thm:generalfor}\\ \hline
Function call & $\Theta(1)$ & 0 & 0 & \ref{thm:fxncall}\\ \hline
\rowheading{Memory Management} \\ \hline
Free lists & $\Theta(N)$ & $\Theta(wN)$ & 0 & \ref{thm:freelist}\\ \hline
Reference Counting & $\Theta(N)$ & $\Theta(wN)$ & 0 & \ref{thm:refcount}\\ \hline
Mark \& Sweep & $\Theta(N)$ & $\Theta(wN)$ & 0 & \ref{thm:marksweep}\\ 

\end{tabular}
\caption{Summary of our reversible primitives analyses and results including control logic, memory management, and garbage collection. In this table, $w$ is the word size, $l$ is the number of loop iteration, and $N$ represents number of memory objects.}
\label{tab:prim-summary}

\end{table}

\begin{table*}[th!]
\centering
\begin{tabular}{l|l|l|l|c}
\rowcolor{blue}
\columnheading{Algorithm} & \columnheading{Time} & \columnheading{Space \unheading (words)}& \columnheading{Energy \unheading (bits)}& \columnheading{Thm.}  \\ \hline
\rowheading{Graph Algorithms} \\ \hline
Breadth-first Search & $\Theta(V+E)$ & $\Theta(V+E)$ & $\Theta(wV +E)$ & \ref{thm:bfs}\\ \hline
Reversible BFS\cite{Frank99} & $\Theta(V+E)$ & $\Theta(V+E)$ & 0 & \ref{thm:bfs-reversible}\\ \hline
Bellman-Ford & $\Theta(VE)$ & $\Theta(V)$ & $\Theta(VEw)$ & \ref{thm:bellman-ford}\\ \hline
Reversible Bellman-Ford & $\Theta(VE)$ & $\Theta(VE)$ & $0$ & \ref{thm:bellman-ford-reversible}\\ \hline
Floyd-Warshall & $\Theta(V^3)$ & $\Theta(V^2)$ & $\Theta(wV^3)$ & \ref{thm:floyd-warshall}\\ \hline
Reversible Floyd-Warshall \cite{Frank99} & $\Theta(V^3)$ & $\Theta(V^3)$ & $0$ & \ref{thm:floyd-warshall-reversible}\\ \hline
Matrix APSP & $\Theta(V^3 \lg V)$ & $\Theta(V^2)$ & $\Theta(wV^3\lg V)$ & \ref{thm:repeated-APSP} \\ \hline
Reversible Matrix APSP \cite{Frank99} & $\Theta(V^3 \lg V)$ & $\Theta(V^2 \lg V)$ & $0$ & \ref{thm:repeated-APSP-reversible} \\ \hline
Semi-reversible Matrix APSP & $\Theta(V^3 \lg V)$ & $\Theta(V^2)$ & $\Theta(wV^2\lg V)$ & \ref{thm:repeated-APSP-reversible} \\ \hline
\rowheading{Data Structures} \\ \hline
Standard AVL Trees (build) & $O(n \lg n)$ & $O(n)$ &  $O(w\cdot n \lg n)$ \\
\phantom{Standard AVL Trees} (search) & $O(\lg n)$ & $O(1)$ & $O(\lg n)$ & \ref{thm:avl-trees-search}\\
\phantom{Standard AVL Trees} (insert) &  $O(\lg n)$ & $O(1)$ &  $O(w \lg n)$ & \ref{thm:avl-trees-insert}\\
\phantom{Standard AVL Trees} ($k$ deletes)  &   $O(k \lg n)$ & $O(1)$ &  $O(w \lg n)$ & \ref{thm:avl-trees-delete}\\ \hline
Reversible AVL Trees (build) & $O(n \lg n)$ & $O(n)$ &  $0$  \\
\phantom{Reversible AVL Trees} (search) & $O(\lg n)$ & $O(1)$ & $0$ & \ref{thm:avl-trees-search-reversible}\\
\phantom{Reversible AVL Trees} (insert) &  $O(\lg n)$ & $O(1)$ &  $0$ & \ref{thm:avl-trees-insert-reversible} \\
\phantom{Reversible AVL Trees} ($k$ deletes)  &   $O(k \lg n)$ & $O(k)$ &  $0$ & \ref{thm:avl-trees-delete-reversible} \\ \hline
Standard Binary Heap
(insert) & $O(\lg n)$ & $O(1)$ & $O(\lg n)$ & \ref{thm:binary-heap-insert}\\
\phantom{Standard Binary Heap} (delete max) & $O(\lg n)$ & $O(\lg n)$ & $O( w \lg n)$ & \ref{thm:binary-heap-delete} \\ \hline
Reversible Binary Heap
(insert) & $O(\lg n)$ & $O(1)$ & $0$ & \ref{thm:binary-heap-insert}\\
\phantom{Reversible Binary Heap} (delete max) & $O(\lg n)$ & $O(\lg n)$ & $0$ & \ref{thm:binary-heap-delete-reversible} \\ \hline
Dynamic Array (build) & $O(n)$ & $O(n)$ &  $0$  \\
\phantom{Dynamic Array} (query) & $O(1)$ & $O(1)$ & $0$ & \ref{thm:dynamic-array}\\
\phantom{Dynamic Array} (add) &  $O(1)$ & $O(1)$ &  $0$ & \ref{thm:dynamic-array} \\
\phantom{Dynamic Array} (delete)  &   $O(1)$ & $O(1)$ &  $0$ & \ref{thm:dynamic-array} \\
\end{tabular}
\caption{Summary of our algorithmic analyses and results. In this table, $n$ is the problem size or number of elements in the data-structure, $w$ is the word size, $\lg$ is $\log_2$, and in graph algorithms, $V$ is the number of vertices, and $E$ is the number of edges.}
\label{tab:algo}
\label{algorithms table}
\end{table*}

\paragraph{Consequences.}
Reducing the energy consumption of many computations by several orders of
magnitude ($n$) will have tremendous impact on practice.
Computer servers alone constitute 23--31 gigawatts of power consumption,
which translates to \$14--18 billion annually and 1.1--1.5\% of
worldwide electricity use \cite{Koomey-2011};
there are roughly 50 times as many PCs with an annual growth rate of 12\%
\cite{Forrester}; and there are about as many smartphones as PCs
\cite{emarketer}.
Improved energy efficiency would save both environmental impact and money.
Reducing energy consumption would also improve the longevity of batteries in
portable devices (laptops, phones, watches, etc.), or enable the
use of smaller and lighter batteries for similar performance.
Perhaps most interestingly, lower energy consumption would lead to faster CPUs,
as cooling is the main bottleneck in increasing clock speeds; reducing the
energy consumption by a factor of $\alpha$, we expect to be able to run the
CPU roughly $\alpha$ times faster.  For example, the world record for
CPU clock speed of 8.429 GHz was set by AMD with liquid nitrogen cooling
\cite{techcrunch}.

Our approach is ambitious in that it requires rethinking both software
(algorithms) and hardware.  Our belief is that building a rich algorithmic
theory for (partially) reversible computation, and showing the orders of
magnitude in possible energy reduction for important problems,
will prove to hardware makers that reversibility is a lucrative feature
worth exploring intensely, even before it becomes inevitable by hitting the Landauer Limit.

\paragraph{Guide.}
This paper has several sections and does not necessarily need to be read in order or in full, depending on the reader's interest. We recommend reading Sections~\ref{sec:WordRAM} and~\ref{sec:Pseudocode} before continuing onto later parts of the paper, to set up the model which is used extensively in the rest of the paper. The remainder of Section~\ref{sec:EnergyModels} further explores our energy models and useful variations. The remaining sections of the paper can be read in any preferred order.  Parts of Sections~\ref{sec:ReversiblePrimitives}--\ref{sec:Algorithms} use results from previous sections, but these should remain understandable without having seen the prior proofs. Section~\ref{sec:ReversiblePrimitives} constructs and analyzes basic control logic and memory management, to enable high-level pseudocode for algorithm specification. Section~\ref{sec:EnergyReductionTechniques} provides some general techniques we have developed for constructing (semi-)reversible algorithms. Sections~\ref{sec:DataStructures} and \ref{sec:Algorithms} analyze several classic algorithms and data structures, and construct new algorithms and data structures that are more energy efficient.
Section~\ref{sec:future} poses open problems.

\section{Energy Models}
\label{sec:EnergyModels}
In the following sections we present three different models of computation which define an energy complexity that attempts to capture the energy loss from Landauer's Principle. We begin with a circuit model due to its intuitiveness and similarity to early work done on reversible logic and computation. We then build up RAM models which bear far more similarity to those used for the analysis of algorithms.

\subsection{Energy Circuit Model}
At the lowest level we will consider logical gates. Every gate is a Boolean function $g:x\rightarrow y$. The energy cost of a gate is defined as the log of the size ratio of the input space, $X$, to the output space, $Y = g(X)$. Thus, energy $E = \lg \left (\frac{X}{Y} \right )$, whose units are bits. The energy cost cannot be negative because a given input cannot map to more than one output. Here we forbid randomized computation. Alternatively, one could allow the creation of $b$ random bits at an energy cost of $b$. Also, the energy cost is zero exactly when the function is bijective in which case we call the gate \emph{reversible}.
\iffull
A \emph{circuit} is a directed acyclic graph of gates with input nodes of in-degree zero whose outputs are the bits corresponding to the input for a problem and whose output is the collection of outputs from the gates which do not go into another gate. The energy cost of a circuit is defined to be the sum of the energy costs of every gate. If the energy cost of a circuit is zero we call it \emph{reversible}.

\begin{lemma}
A reversible circuit must compute a bijective function.
\end{lemma}

\begin{proof}
For a circuit to be reversible, all of its gates must be reversible, and thus bijective. The composition of bijective functions is also bijective.
\end{proof}

We now have a framework for understanding the energy complexity of circuits. Although we imagine there are very interesting things that can be said about the energy complexity of circuits, we primarily want this infrastructure to help justify pieces of further models. For that purpose, general circuits are too powerful and we wish to specify a class of circuits which are simple enough they may reasonably make up the logic in a computer. Normally one restricts to AND, OR, and NOT gates; however, this would make it impossible to construct most circuits reversibly. Luckily, there are small, universal, reversible logic gates. Thus, we will also permit the use of Fredkin gates, also known as controlled-swap gates, and Toffoli gates, also known as controlled-controlled-not gates. Further, we allow the (reversible) zero-input gates \textsc{true} and \textsc{false} whose outputs are 1 and 0 respectively; as well as \textsc{endtrue} and \textsc{endfalse} whose inputs are restricted to be 1 and 0 respectively and whose output is the null set. Now for a given function $f$, we can define a new circuit-size complexity as the minimum number of \textsc{and}, \textsc{or}, \textsc{not}, \textsc{true}, \textsc{false}, \textsc{endtrue}, \textsc{endfalse} , Toffoli, and Fredkin gates needed to create a circuit which computes $f$. Similarly each circuit has a circuit-depth, corresponding to the longest path in the circuit, and thus we can similarly define a circuit-depth complexity for $f$. 

Fredkin and Toffoli give a universality result, showing that any reversible function can be computed by reversible circuits, given extra or ancillary bits. Further, only a constant-factor increase in circuit-depth and complexity is needed over a similar circuit constructed from \textsc{and}, \textsc{or}, and \textsc{not} gates \cite{conservative82}. In a recent paper, Aaronson, Grier, and Schaeffer go further classifying all reversible logic gates on bits. They also prove any reversible function can be implemented with only a constant number of ancillary bits \cite{AaronsonGS15}. To give some intuition, Figure~\ref{ToffoliLogic} shows how Toffoli gates can be used to create other logical operations. 

\begin{figure}[!ht]
  \centering
  \begin{subfigure}[b]{0.19\textwidth}
    \includegraphics[width=\textwidth]{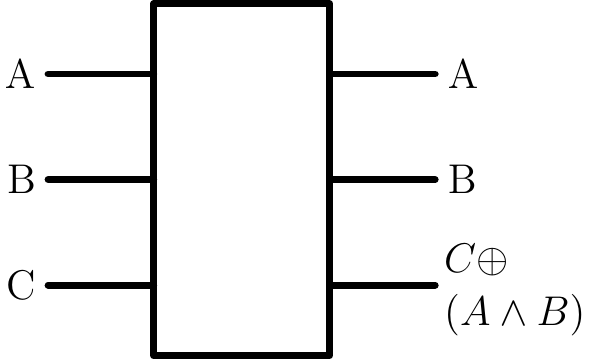}
    \caption{Toffoli gate}
    \label{Toffoli}
  \end{subfigure}
  \hfill
  \begin{subfigure}[b]{0.19\textwidth}
    \includegraphics[width=\textwidth]{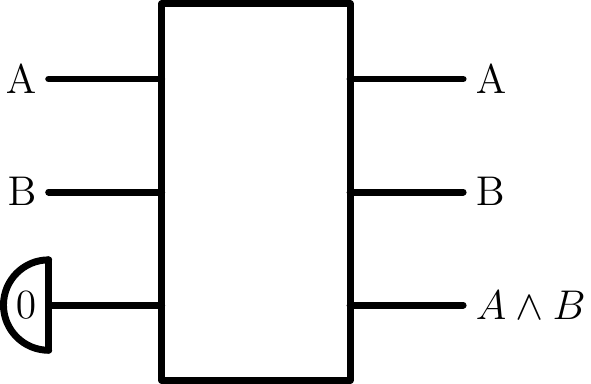}
    \caption{\proc{And}}
    \label{And}
  \end{subfigure}
  \hfill
  \begin{subfigure}[b]{0.19\textwidth}
    \includegraphics[width=\textwidth]{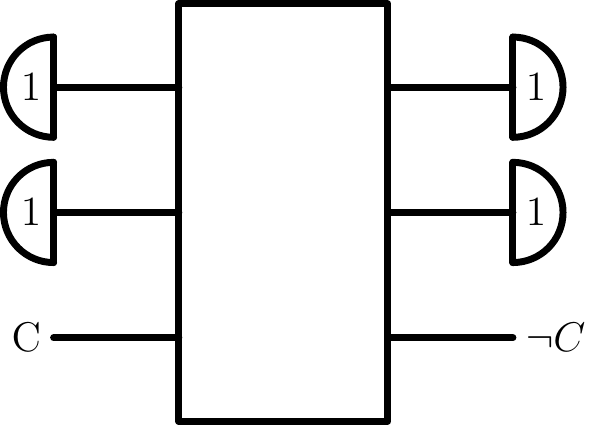}
    \caption{\proc{Not}}
    \label{Not}
  \end{subfigure}
  \hfill
  \centering
  \begin{subfigure}[b]{0.19\textwidth}
    \includegraphics[width=\textwidth]{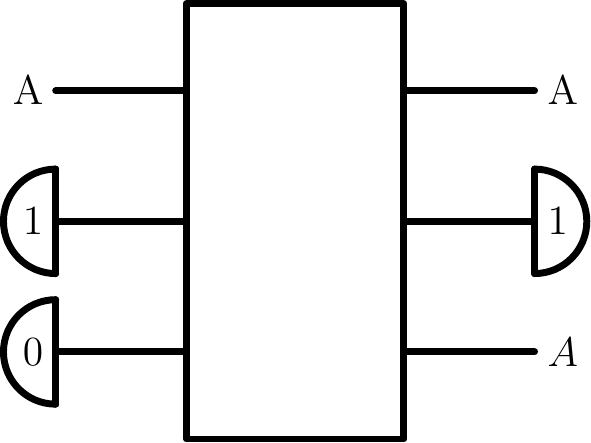}
    \caption{\proc{Copy}}
    \label{Copy}
  \end{subfigure}
  \hfill
  \begin{subfigure}[b]{0.19\textwidth}
    \includegraphics[width=\textwidth]{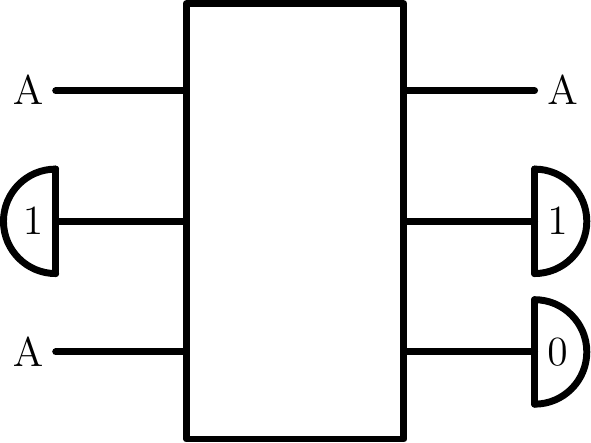}
    \caption{\proc{DeleteCopy}}
    \label{DeleteCopy}
  \end{subfigure}
  \caption{Standard logic implemented using a Toffoli gate.}
  \label{ToffoliLogic}
\end{figure}

Because Fredkin and Toffoli gates can similarly be simulated by a constant number of AND, OR, and NOT gates, this should not impact circuit complexity beyond constant-factors; however, it is necessary to be able to capture our notion of energy complexity in a useful way.
\fi

\subsection{Energy Word RAM Model}
\label{sec:WordRAM}
The Energy Word RAM model allows any contiguous segment of memory of size $w$ to be accessed in constant time and defines a fixed set of operations that can take in $O(1)$ word sized inputs in constant time. We also assume memory allocation is handled in a reversible manner. This will become a more reasonable assumption later, when we show linked-lists and stacks can be implemented reversibly. The program and operations have the following restrictions. First, we restrict ourselves to the operations typically found in high-level languages as well as their reversible analogues. Second, the operation's energy costs should be calculated based off of what can be constructed in the circuit model. Third, all reversible operations must come paired with their inverse operation. Finally, all Energy Word RAM programs must return the machine to its original state, with the exception of a copy of the output living somewhere in memory. This can be done simply but expensively by irreversibly zeroing out every bit and paying the associated energy cost.

The reversible operations we allow include in-place addition and subtraction
(e.g., $a \mathrel+= b$), increment and decrement (e.g., $a \mathrel+= 1$), swapping two
variables, testing for equality or less-than relation, copying a variable
into an initially empty variable
$$\left (\proc{Copy}(a, \underline{b}) \equiv b \mathrel+= a \right )$$ and destroying a known
copy of a variable $$\left (\proc{DestroyCopy}(a,b) \equiv b \mathrel-= a \right )$$ We have introduced here a useful notation, that of underlining variables whose values are empty, which shall serve us in writing pseudocode as well. The irreversible operations we allow include overwriting one variable
with another, and computing the bitwise \proc{And} or \proc{Or} of two variables.

\iffull
\subsubsection{Irreversible Word Operations}
\begin{itemize*}
\item $\proc{And}(a,b)$: Returns bitwise And. Costs $1$ unit of energy.
\item $\proc{Or}(a,b)$: Returns bitwise Or. Costs $1$ unit of energy.
\item $\proc{Set}(a,b)$: Sets the value of $a$ equal to the value of $b$. Costs $w$ energy.
\end{itemize*}

\subsubsection{Reversible Word Operations}
\begin{itemize*}
\item $\proc{Add}(a,b)$: Returns $a+b$ and $b$. Inverted by $\proc{Sub}(a,b)$
\item $\proc{Sub}(a,b)$: Returns $a-b$ and $b$. Inverted by $\proc{Add}(a,b)$
\item $\proc{Inc}(a)$: Returns $a+1$. Inverted by $\proc{Dec}(a)$
\item $\proc{Dec}(a)$: Returns $a-1$. Inverted by $\proc{Inc}(a)$
\item $\proc{Not}(a)$: Assumes $a$ is 0 or 1. Returns the logical negation of $a$. Inverted by $\proc{Not}(a)$.
\item $\proc{Swap}(a,b)$: Swaps the values stored in the memory locations of $a$ and $b$. Inverted by $\proc{Swap}(a,b)$.
\item $\proc{LessThanOrEqualTo}(a,b)$: Returns the inputs and 1 if $a\leq b$ and 0 otherwise. Note, if the extra output bit is deleted after it is used, it will cost $1$ energy.
\item $\proc{InverseLessThanOrEqualTo}(a,b,c)$: Returns $a$, $b$, and $c-1$ if $a\leq b$. Returns $a$, $b$, and $c$ if $a>b$.
\item $\proc{Copy}(a, \underline{b})$: Returns $a$ and $b+a$. If $b$ is known to be zero, this results in a successful copy. Inverted by $\proc{DestroyCopy}(a,b)$.
\item $\proc{DestroyCopy}(a,b)$: Returns $a$ and $b-a$. If $b$ is known to be equal to $a$, then $b$ is now zero. Inverted by $\proc{Copy}(a,b)$.
\end{itemize*}
\fi

In this model, we intend that our lowest level pseudocode correspond to an assembly-like language. For simplicity we will continue to work with variables and locations in memory as though they are all stored in RAM, rather than deal with registers, paging, and other complications that may arise depending on the computer architecture. At this level we also explicitly number every line of our program and grant the code access to the program counter, PC, which is the location in memory of the current instruction. At every instruction the PC is incremented, but it can also be manipulated manually, allowing jumps among other operations. It is very easy to make code irreversible by manipulating the PC, as this is implicitly adding control logic to the program. The instruction set we will be using in this paper is the same as the one with which we defined our Word RAM model. For an instruction set for a reversible computer that has been built see Appendix B of Frank's Thesis \cite{Frank99} or \cite{BobISA}.

\subsection{Energy Transdichotomous RAM Model}
The Energy Transdichotomous RAM model is computationally the most powerful and flexible. As with the Word RAM model, we allow access to memory segments of size $w$ in constant time and assume memory allocation is done reversibly. Generally we will assume that $w=\Omega(\lg n)$, making the word size capable of indexing the entire input of the problem. We also allow any operations on $O(1)$ words to be performed in constant time; however, every algorithm can only use a constant number of different operations. The energy cost of an operation is simply the log of the ratio of the input space to the output space, as in the circuit model. Note that this is a lower bound on the energy cost of the operation in the circuit model, and thus a lower bound in the Energy Word RAM model. Finally, we still need to leave the computer in its initial state, except for a copy of the output.

This model is convenient to work in because it is relatively easy to calculate the energy cost of many operations and the flexibility of choosing operations allows us to exploit information in the system without having to work out the details of how it would be implemented. For example, when dividing an integer by four would generally incur two bits of energy loss or two bits of garbage; however, if we happen to know that the number is even, there is really only a single bit of information being lost. Instead of having to worry about how to perform shifts and additions to save this bit, the Transdichotomous RAM model allows us to have a `divide by four when evenly divisible by 2' operation with the restriction that it only takes even inputs.

We now develop some conventions for writing programs in the Transdichotomous RAM Model. All lines are of the form $\id{TUPLE} = \id{TUPLE}$. Both tuples must contain the same number of elements, and the number of elements must be $O(1)$. The left tuple is a list of all of the values in memory which are used in the computation being performed on this line, including those simply being overwritten. The right tuple contains expressions representing the values that will be in the corresponding variables on the left. These expressions must contain no more than $O(1)$ constant time operations. One interesting convention about this language is every variable implicitly serves two purposes depending on its location. On the left, all variables refer to the memory location where they are stored, and on the right they refer to the values being represented at those memory locations.

As we did above, here we shall annotate variables whose value is known to be zero (often new, unassigned variables) with an underline. This information is often critical to the energy cost of an expression. For example, $(a, b, c) = (a, b, a+b)$ would cost $w$ units of energy because we are erasing every bit in $c$ before replacing it with the value $a+b$. However, $(a,b,\ul{c}) = (a,b,a+b)$ has no energy cost because the input has the value of $c$ assumed to be zero, thus reducing the input space by a factor of $2^w$ and making the number of inputs and outputs the same.

The following are some examples of common operations written in the format. All operations are assumed to be integer operations with reasonable overflow and rounding conventions. 
The following examples cost zero energy:
\begin{itemize*}
\item \proc{Copy}: $(a, \ul{b}) = (a,a)$
\item \proc{DestroyCopy}: $(a, b) = (a, b-a)$
\item \proc{Add}: $(a, b) = (a+b, b)$
\iffull
\item \proc{Add} a constant: $(a)=(a+5)$
\item Quotient and remainder of two numbers: $(a, b, \ul{c}) = (a/b, a\bmod{b}, b)$
\item CNOT (aka reversible XOR): $(a, b) = (a, a \oplus b)$
\item AND: $(a, b, \ul{c}) = (a, b, a \wedge b)$
\item Multiplication: $(a, b, \ul{c}) = (a, a*b >> w, a*b \bmod{2^w})$
\fi
\item \proc{LessThan}: $(a, b, \ul{c}) =(a,b,a<b)$ 
\end{itemize*}
\iffull
The following examples cost $w$ units of energy:
\begin{itemize}
\item \proc{Delete}: $(a) = (0)$
\item \proc{AND}: $(a,b) = (a \wedge b, 0)$
\item \proc{XOR}: $(a,b) = (a \oplus b, 0)$
\item \proc{Add}: $(a,b) = (a + b, 0)$
\end{itemize}
Here are some examples of common operations that cost between $0$ and $w$ units of energy:
\begin{itemize}
\item \proc{AND}: $(a,b) = (a \wedge b, b)$. This has an energy cost of $1$
\item Right shift by 3: $(a) = (a << 3)$. This has an energy cost of $3$.
\item Remainder: $(a, b) = (a \bmod{b}, b)$. This has an energy cost of $w - \lg b$
\end{itemize}
\fi
\iffull
We introduce another convention to avoid writing down every variable
that is not changing:
If a variable appears in the right-hand side of a line and not the left it is short-hand for that value being assigned to the same memory location it originally came from. We will also omit parenthesis if there is only one variable or expression. Thus, $(a,b,c) = (a+b+c,b,c)$ becomes $a=a+b+c$ and $(a, b) = (b, b)$ becomes $a = b$. 
\fi
\iffull
One may initially worry that this will obscure the energy cost of the operation; however, it actually makes no difference. If a variable appears as a single element of the input and output tuple, it contributes the same factor to the input space as the output space. These will cancel when calculating the entropy of the operation and thus omitting variables which are assigned back to themselves will not change the calculated energy cost. We encourage the use of this convention to make the code easier to understand.
\fi

\subsection{High-level Pseudocode}
\label{sec:Pseudocode}
Although the previous section provides a nice, clean way to analyze the energy, space, and time complexity of an algorithm; we may want a more concise and C-like language. Past research on reversible programming languages has focused on fully reversible programming languages and architectures. The first high-level reversible programming languages developed were Janus \cite{lutz1982janus}\cite{yokoyama2010reversible} and R \cite{Frank99}. The first reversible architecture, Pendulum, was developed by Vieri \cite{vieri1999reversible}\cite{vieri1998fully}. Along with Pendulum, Vieri introduced a reversible low-level instruction set, PISA, which is used as a basic reversible instruction set for many future works. Most recently, this architecture has been further improved with the development of Bob \cite{BobISA} using a slightly modified version of PISA known as BobISA, providing more efficient branch handling and address calculation. Axelsen \cite{axelsen2011clean} presented the first compilation techniques to translate high-level Janus to low-level PISA, two independently developed reversible languages, and showed that his techniques can be extended for use in any high-level reversible language.

We modeled our pseudocode off of these previous high and low level reversible languages while also adding a few new commands to allow for partial reversibility. We now allow lines of the form $\id{VARIABLE}=\id{EXPRESSION}$ as well as for loops, while loops, if/else statements,  and subroutine calls. We also introduce log blocks and unroll statements in Section~\ref{sec:LoggingAndUnrolling}. On lines where we are assigning a variable, we assume that every input in the expression will remain unchanged in its memory location after the computer performs the operation and that the variable on the left-hand side will have its value replaced by the value of the expression. If this is a reversible operation, the variable will merely be changed as appropriate; if it is an irreversible operation, then the variable will be changed and an additional energy cost will be incurred based on the model being used. 

\iffull
Figures~\ref{example 1} and~\ref{example 2} give some simple examples of
equivalent code in the three different levels of pseudocode conventions we've developed (high, intermediate, low). The high level is our C-like language. The intermediate language converts high level control logic to jumps and labels. The low level breaks it down further to an assembly-like language. Future sections
will use one or more conventions as needed for clarity.
\fi

\ifabstract
Figure~\ref{example 1} gives some simple examples of
equivalent code in the three different levels of pseudocode conventions we've developed (high, intermediate, low). The high level is our C-like language. The intermediate language converts high level control logic to jumps and labels. The low level breaks it down further to an assembly-like language. Future sections
will use one or more conventions as needed for clarity.
\fi

\begin{figure}
\begin{minipage}{0.45\linewidth}

\begin{pcode}
$x$ = $x+y+z$ \lasttab \textbf{high}
\end{pcode}
\begin{pcode}
$(x, y, z)$ = $(x+y+z, x, z)$ \lasttab \textbf{intermediate}
\end{pcode}
\begin{pcode}
101  \id{\ul{tempx}} = $x$ \lasttab \textbf{low}
102  $x$ += $y$
103  $y$ $-$= $x$
104  $x$ += $z$
105  $y$ += \id{tempx}
106  $y$ += \id{tempx}
107  \id{tempx} $-$= $y$
\end{pcode}
\caption{Simple example of code in high, intermediate, and low-level pseudocode.}
\label{example 1}
\end{minipage}
\hfil\hfil
\iffull
\begin{minipage}{0.45\linewidth}
\begin{pcode}
\keyw{log} \underline{x} = 5
\keyw{log} x = 10
\keyw{unroll}

$\ul{x} = 5$
$(\ul{\id{tempx}}, x) = (x, 0)$
$\ul{x} = 10$
$x = x-10$
$(\id{tempx}, \ul{x}) = (0,\id{tempx})$
$x = x-5$

101  $x$ += 5
102  \id{mem}[\id{lp}] += $x$
103  $x$ $-$= \id{mem}[\id{lp}]
104  \id{lp} += 1
105  $x$ += 10
106  $x$ $-$= 10
107  \id{lp} $-$= 1
108  $x$ += \id{mem}[\id{lp}]
109  \id{mem}[\id{lp}] $-$= $x$
110  $x$ $-$= 5
\end{pcode}
\caption{Simple example of logging in high, intermediate, and low-level pseudocode.}
\label{example 2}
\end{minipage}
\fi
\end{figure}

\subsubsection{Logging and Unrolling}
\label{sec:LoggingAndUnrolling}
Dealing with garbage data tends to become tedious when writing reversible computer code. For example, suppose that we were comparing two variables, $a$ and $b$, and that we wanted to use the result of this comparison to increase some counter; see Figure~\ref{logunroll-a}.

\begin{figure}
\centering
\begin{subfigure}{0.22\linewidth}
\begin{pcode}
\ul{$a$} = $x > y$
\id{counter} += $a$
\end{pcode}
\caption{garbage data not unrolled}
\label{logunroll-a}
\end{subfigure}
\hfil\hfil
\begin{subfigure}{0.22\linewidth}
\begin{pcode}
\keyw{log}: \+
\ul{$a$} = $x > y$ \-
\id{counter} += $a$
\keyw{unroll}
\end{pcode}
\caption{logged high-level}
\label{logunroll-c}
\end{subfigure}
\hfil\hfil
\begin{subfigure}{0.22\linewidth}
\begin{pcode}
\ul{$a$} = $x > y$
\id{counter} += $a$
$a$ $-$= $x > y$
\keyw{dealloc}($a$)
\end{pcode}
\caption{logged low-level / automatic unroll}
\label{logunroll-d}
\end{subfigure}
\caption{Three examples detailing the mechanics of logged code. }
\end{figure}

In a normal computer, by the function's end, $a$ would be garbage-collected automatically; however, in our reversible computer a naive garbage collection algorithm would destroy the information stored in $a$, clearing whatever value it held and costing a word of energy. Thus, the reversible algorithm programmer must handle the task of deallocating $a$ manually.

We call the process of using a series of commands to directly reverse some portion of the code \emph{unrolling}. Manually writing all such commands can be tedious and is prone to error.  To expedite the process, we introduce the high-level keywords \keyw{log} and \keyw{unroll}:

In Figure~\ref{logunroll-c}, the line {\ul{$a$} = $x > y$} is included inside the \keyw{log} indentation block, and so is to be reversed at the call to \keyw{unroll}. For much longer programs, this extra syntax can save the programmer a great deal of effort that would otherwise be spent writing reverse code. Note that the \keyw{log} and \keyw{unroll} commands only exist in the highest-level language, and are translated into their manual equivalent at compile-time. The above program, therefore, would compile to the low-level program seen in Figure~\ref{logunroll-d}.

The rules for unrolling are straightforward. Reversible commands can be unrolled simply by including their inverse commands in reverse calling order. Unrolling reversible control logic is discussed in Section~\ref{sec:control-logic}.

To allow our model to unroll semi-reversible programs, which may include irreversible commands, we introduce the \emph{log stack}, a data structure onto which the program can push extra bits of information to be used later to invert the otherwise-irreversible operations. We keep track of our position in the log stack with the log pointer, $lp$. In the Transdichotomous model, every operation must have its inverse and the process for logging that operation explicitly specified. Furthermore, we assume that this garbage is encoded as efficiently as possible and thus only requires as many bits of space as are needed to distinguish the input space from the output space. Once again, when we log lines with operations that were previously irreversible, we are implicitly defining new operations and should take appropriate precautions. In our Word RAM model, these operations, their inverses, and operations capable of interfacing with $lp$ and memory must be specified.

\subsubsection{Promise Notation}

We introduce another notational convention that will assist in writing low-energy pseudocode for the Transdichotomous RAM model. At the end of a standard line of code, one may add a comma, the keyword ``assert'', and then a claimed Boolean expression restricting the values of the involved variables. Some useful examples include: \\

\begin{pcode}
\qquad $\id{IsTrue} = 0,$ \keyw{assert} $0 \le \id{IsTrue} \le 1$
\qquad $x = x/y,$ \keyw{assert} $6 \mid x$
\end{pcode}

Here $\id{IsTrue}$ may have been the result of a comparison and is known to be either $0$ or $1$. Thus the energy cost of destroying it is only $1$ bit instead of $w$ bits. In the second example, we might know that the problem being computed has some symmetries that a compiler might not see which restrict the values $x$ can take on. Asserts allow us to implicitly define functions which have a restricted input space and thus reduce energy costs. Given the convenience of defining functions in this manner, we must be very careful that we are still using only $O(1)$ different operations in our algorithm.

\section{Reversible Primitives} \label{sec:ReversiblePrimitives}

In this section, we develop many high-level primitives commonly used
throughout algorithms, but which need special care to be done in an
energy-efficient manner. Before proceeding, we should discuss in slightly more
detail the architecture of our theoretical semi-reversible computer.
Our computer only has a single mode of operation, always incrementing
the program counter with every instruction. Reversing operations comes 
from writing the inverse operation in a later section of code, rather than
having a separate reversal mode which travels backward along the 
program counter, inverting those operations. This gives us more flexibility
in how to handle irreversible sections of code, and the manner in which
we reverse operations which are not dependent upon each other. However,
it comes at the cost that we cannot recover the value of the program 
counter. Thus this design will have to incur an energy cost of $w$ every
time the computer is reset. Because we can run many programs between
restarts, we do not consider this to be of major consequence.

\subsection{Control Logic}
\label{sec:control-logic}

Following Frank \cite{Frank99}, we can make branching logic reversible with constant space overhead using paired branching with the destination of a branch being a branch that points back. Thus, we have symmetry, and when running backward, we can just follow the branch we arrived on. However, because we are not working in a fully reversible model, there are  some caveats we must pay attention to: all reversible control logic in this section depends on all of the code within the control sequence being reversible. If this is not the case, we can make no guarantees about the correctness when irreversible operations are being performed within some control logic, especially if they are manipulating the variables the control logic depends on.

In Section~\ref{sec:WordRAM} we noted that we can do comparisons reversibly with a single bit of extra space. In this section, we look at jumps, branches, conditionals, for loops, and function calls.
\iffull
\subsubsection{Jumps and Branches}
\fi
Here we consider the most basic building blocks of control logic, alterations to the program counter in the form of jumps and branches (conditional jumps). Jumps can be performed by a reversible addition to the program counter, we use notation \keyw{goto}, \keyw{gotoifeq}, and \keyw{gotoifneq}. However, if the program counter is allowed to change, we can no longer assume every line was reached by an increment to the program counter, thus creating an irreversible situation. To deal with this, all program counter jumps must be paired with a \keyw{comefrom}, \keyw{comefromifeq}, or \keyw{comefromifneq} statement. In our pseudocode, we allow for \keyw{goto} to direct an absoloute or relative jump and note that a compiler can transform absolute to relative, as is used in most reversible architectures.

\begin{theorem}\label{thm:jump}
Jumps can be implemented reversibly with constant-factor increases in time and space and up to an extra word of space per jump.
\end{theorem}

\begin{proof}
All jumps must be paired with \keyw{comefrom} statements. In the case of a regular \keyw{comefrom} statement, the program knows that it reached this location via a jump and thus will jump back in the reverse. However, it is rare that an unconditional jump such as this exist. In the more general case, the program must decide whether the comefrom was reached via a jump or from the line above by an increment in the program counter. To address this, we log two things upon jumping: (1) the length of the jump and (2) a bit indicating that a jump occurred. We then use a \keyw{comefromif} to check this bit upon reversing. At the corresponding reverse \keyw{comefromif} we'll pop the value off the log stack and use it to either change the program counter or not depending on whether the code jumped to that location. \iffull A general example is shown in Figure~\ref{jump}.\fi

We can make an additional optimization if a \keyw{comefrom} statement only has one corresponding \keyw{goto}. Because we know the jump location corresponding with the \keyw{comefrom}, this can be implemented reversibly by noting the jump length directly in the source code and not logging. In this case we only have a single bit of storage for logging whether the jump was taken. \iffull This is illustrated in Figure~\ref{jumppair}.\fi
\end{proof}
\iffull
We make the distinction above between two flavors of jumps. The first, paired jump, requires that \keyw{goto} and \keyw{comefrom} be in a $1:1$ pairing, meaning the jump length is pre-determined and only a single bit of storage is required. The second, variable jump, occurs when the jump length is determined is determined at run-time. This can occur when, for example, multiple jumps land at the same destination\iffull (\ref{jump})\fi. This requires pulling the destination off the log and takes an additional word of space.

\begin{figure}
\begin{pcode}
\keyw{log}: \+
  \mbox{[code1]}
  \keyw{goto} \id{label1}
  \mbox{[code2, $k$ lines long]}
  \keyw{label} \id{label1}
  \mbox{[code3]}  \-
\mbox{[code4]}
\keyw{unroll} 
\end{pcode}

\begin{pcode}
\mbox{[code1]}
\id{mem}[\id{lp}] += 1
\id{lp} += 1
\keyw{goto} $k$
\mbox{[code2]}
\id{lp} += 1 \lasttab{//adds 0 to log indicating no jump}
\mbox{[code3]}
\mbox{[code4]} \lasttab{//unroll begins}
\mbox{[reverse of code3]}
\keyw{comefromifeq}(\id{mem}[\id{lp}], $k$)
\id{lp} -= 1 \lasttab{//removes 0 bit from log}
\mbox{[reverse of code2]}
\id{lp} -= 1
\id{mem}[\id{lp}] -= 1 \lasttab{//removes 1 bit from log}
\mbox{[reverse of code1]}
\end{pcode}

\caption{The intermediate-level and compiled low-level code for a reversible jump. The jump and comefrom statements are paired uniquely so the jump distance can be placed directly in the low level code.}
\label{jumppair}
\end{figure}
\fi

\iffull
\begin{figure}
\begin{pcode}
\keyw{log}: \+
  \mbox{[code1]}
  \keyw{goto} \id{label1} \lasttab{//call this jump1}
  \mbox{[code2]}
  \keyw{goto} \id{label1} \lasttab{//call this jump2}
  \mbox{[code3]}
  \keyw{label} \id{label1}
  \mbox{[code4]} \-
\mbox{[code5]}
\keyw{unroll}
\end{pcode}

\begin{pcode}
\mbox{[code1]}
\id{mem}[\id{lp}] += \mbox{[distance from jump1 to label1]}
\id{lp} += 1
\id{mem}[\id{lp}] += 1
\id{lp} += 1
\keyw{goto}(\id{mem}[\id{lp-2}]) \lasttab{//this is jump1}
\mbox{[code2]}
\id{mem}[\id{lp}] += \mbox{[distance from jump2 to label1]}
\id{lp} += 1
\id{mem}[\id{lp}] += 1
\id{lp} += 1
\keyw{goto}(\id{mem}[\id{lp-2}]) \lasttab{//this is jump2}
\mbox{[code3]}
\id{lp} += 1 \lasttab{//adds 0 to log indicating no jump occurred}
\keyw{comefromifeq}(\id{mem}[\id{lp-1}], \id{mem}[\id{lp-2}]) \lasttab{//jump occurred, jump length}
\mbox{[code4]}
\mbox{[code5]} \lasttab{//unroll begins}
\mbox{[reverse of code4]}
\keyw{gotoifeq}(\id{mem}[\id{lp-1}], \id{mem}[\id{lp-2}]) \lasttab{//reverse of comefromifeq}
\id{lp} -= 1
\mbox{[reverse of code3]}
\id{mem}[\id{lp}] $-$= 1 \lasttab{//reverse of jump2}
\id{lp} $-$= 1
\id{mem}[\id{lp}] -= \mbox{[distance from jump2 to label1]}
\id{lp} $-$= 1
\mbox{[reverse of code2]}
\id{mem}[\id{lp}] $-$= 1 \lasttab{//reverse of jump1}
\id{lp} $-$= 1
\id{mem}[\id{lp}] -= \mbox{[distance from jump1 to label1]}
\id{lp} $-$= 1
\mbox{[reverse of code1]}
\end{pcode}
\caption{High-level and compiled low-level code for a general reversible jump. The jump length is stored so that the reverse jump can be made.}
\label{jump}
\end{figure}
\fi

\subsubsection{Conditional Statements}

We distinguish between two different types of conditionals, a \emph{protected} \keyw{if} statement and a general \keyw{if} statement. A protected \keyw{if} statement is one in which the conditional is not modified within the \keyw{if} statement.

\begin{theorem}\label{thm:if}
Protected If statements can be implemented reversibly with constant-factor increases in time and space, and general If statements with an extra bit of overhead in space.
\end{theorem}

\iffull
\begin{proof}

Let $a$ represent the variable which holds the conditional expression of the \keyw{if} statement. First consider a protected statement. Regardless of whether $a$ was false and we jumped or $a$ was true and we entered the \keyw{if} statement body, $a$'s value will be preserved at the end of the \keyw{if} statement. As described in \cite{Frank99}, two-way branching may then be employed to make the \keyw{if} statement reversible. 

To briefly summarize Frank's results, in the forward direction, the value of $a$ determines whether we enter the body of the \keyw{if} statement or jump to the end. In a similar way, when reversing the operation, if we leave a marker at the end of the \keyw{if} statement indicating that this location may have been reached via two different paths --- either by jumping or by flowing through the \keyw{if} statement body --- depending on the value of $a$. Because $a$ will be maintained, we can use a comefromif conditioned on $a$ to determine if the \keyw{if} statement body should be undone. Because this \keyw{if} statement is reversible, it does not incur any cost in our logging space (Figure \ref{ifreversiblelogged}). 

In a general \keyw{if} statement, because $a$ is not protected and is subject to change, we must log the value of $a$ upon initial execution and use the logged value in the comefromif to determine whether to jump in the reverse direction. This gives an extra bit of storage.

\begin{figure}
\begin{minipage}{0.45\linewidth}
\begin{pcode}
\keyw{ifp}($x > y$):  \+ 
   $z$ $+$= 5 \-
$x$ $-$= 6
\end{pcode}
\begin{pcode}
101  \ul{$a$} = $x > y$
102  \keyw{gotoifneq}(a, 12)
103  $z$ $+$= 5
     \ldots
114  \keyw{comefromifneq}(a, 12)
115  $a$ $-$= $x > y$
116  $x$ $-$= $6$
\end{pcode}
\caption{A reversible \keyw{if} statement in high-level and low-level pseudocode. The \keyw{comefrom} statement on line 114 is a marker which, though it does not explicitly cause any change to the state of the program during forward execution, is essential to ensure the instantaneous reversibility of the code.}
\label{ifreversible}
\end{minipage}
\hfil\hfil
\begin{minipage}{0.45\linewidth}
\begin{pcode}
\keyw{log}: \+
    \keyw{ifr}($x > y$): \+
       $z$ $+$= 5 \-
    $x$ $-$= 6 \-
$y$ $+$= $x$
\keyw{unroll}
\end{pcode}
\begin{pcode}
101  \ul{$a$} = $x > y$
102  \keyw{gotoifneq}(a, 2)
103  $z$ $+$= 5
104  \keyw{comefromifneq}(a, 2)
105  $a$ $-$= $x > y$
106  $x$ $-$= 6
107  $y$ $+$= $x$ //unroll begins
108  $x$ $+$= 6
109  $a$ $+$= $x > y$
110  \keyw{gotoifneq}(a, 2) //reverse of comefrom
111  $z$ $-$= 5
112  \keyw{comefromifneq}(a, 2)
113  $a$ $-$= $x > y$
\end{pcode}
\caption{Logged reversible \keyw{if} statement in high-level code and compiled low-level code}
\label{ifreversiblelogged}
\end{minipage}
\end{figure}

\end{proof}

\subsubsection{For Loops}
\label{sec:forloops}
\fi

We now examine a special case of \keyw{for} loops. A \emph{simple for loop} is one in which a variable $i$ iterates over the values $1$ through $k$, each time executing some piece of code which does not alter $i$ or $k$. 

\begin{theorem}\label{thm:simplefor}
Simple for loops can be implemented reversibly with constant-factor increases in time and space.
\end{theorem}
\iffull
\begin{proof}
At the end of the body of the \keyw{for} loop we can simply increment $i$, check if it is less than or equal to $k$, and if not use a branch to jump back to the beginning of the body of the for loop (Figure \ref{simplefor}). Incrementation can be done reversibly and we have shown branching can be done reversibly. It is easy to see one can unroll one of these loops by simply running the inverse code for $i$ ranging from $k$ to $1$. Because $i$ only exists within the \keyw{for} loop and we know the loop was executed, we do not need to worry about whether or not to execute the inverse code, or how many loops were performed.
\end{proof}

\begin{figure}
\begin{pcode}
\keyw{log}: \+
   \keyw{for}($\ul{i}$=1 ; $i \ge 1$ ; $i <x$ ; $i$++): \+
      \mbox{[code1]} \-
   \mbox{[code2]} \-
\keyw{unroll} 
\end{pcode}

\begin{pcode}
\keyw{log}: \+
$\ul{i}$ = 0
\keyw{label} \id{beginFor}
$i$ += 1
\keyw{gotoifneq} ( $i$ $<$ $x$, \id{endFor}):
\mbox{[code1]}
\keyw{goto} \id{beginFor}
\keyw{label} \id{endFor}
\mbox{[code2]}\-
\keyw{unroll} 
\end{pcode}

\caption{An example of a simple for loop being translated from high-level to intermediate-level code consisting of jumps and branches.}
\label{simplefor}
\end{figure}
\fi

We consider an extension of the simple for loop. A for loop has \emph{internal conditions} if all variables used in the condition of the for loop only exist within the scope of the for loop. That value is \emph{protected} if it is never changed irreversibly. We define a \emph{protected} for loop as a for loop with protected, internal conditions\iffull (Figure \ref{protectedfor})\fi.

\begin{theorem}\label{thm:protectedfor}
A protected for loop can be performed reversibly with constant factor overhead in time and space.
\end{theorem}

\iffull
\begin{proof}
This follows the same reasoning as the proof of Theorem~\ref{thm:if}. We can construct the for loop out of a paired branch statement. Since we know the terminal value of the iterator is not changed after the loop ends, then when reversing, we can check it's value to see if the for loop ever executed. The update function on the iterator is also promised to be reversible. Thus to determine if one is done unwinding the loop, one can apply the inverse update to the iterator and then check whether it has reached its initial value.
\end{proof}

\begin{figure}
\begin{pcode}
\keyw{log}: \+
\mbox{[code1]}
\keyw{for} ($\underline{i}$ = $a$; $g(i, b)$; $f(i)$): \+
	\mbox{[code2]} \-
\mbox{[code3]} \-
\keyw{unroll}
\end{pcode}

\begin{pcode}
\keyw{log}: \+
\mbox{[code1]}
$\ul{i}$ = $f^{-1}(a)$
\keyw{label} \id{beginFor}
$f(i)$ \lasttab{//$f$ is a reversible function}
\keyw{gotoifneq} ( $g(i,b)$, \id{endFor}):
\mbox{[code2]}
\keyw{goto} \id{beginFor}
\keyw{label} \id{endFor}
\mbox{[code3]}\-
\keyw{unroll} 
\end{pcode}
\caption{Example of high-level to intermediate-level code for a protected for loop}
\label{protectedfor}
\end{figure}

In general, loops cannot be performed reversibly without overhead. For example, we may not be able to distinguish whether or not a loop ever executed because the variable checked in the condition may have gone outside the range during the loop or have been set that way before the loop ever executed. Even if we know a loop has executed, it is possible for the variables considered by the loop to reach the same initial value multiple times throughout the execution of the loop. If this is the case, then we cannot rely on these initial values to tell us to terminate the loop, as was the case with simple for loops. 

However, this issue can be addressed with a small amount of overhead. We will simply create an extra variable at the start of the loop and increment it whenever we loop. Once the loop terminates, we will push this variable onto the log stack, so we can recover the number of times our loop executed. Assuming the internals of the loop are reversible, then the entirety of the loop is now reversible. We call this a \emph{general} for loop.
\fi

\begin{theorem}\label{thm:generalfor}
A general for loop can be performed reversibly with constant factor overhead in time and an extra word of space representing the number of loop executions.
\end{theorem}

\iffull
\subsubsection{Function calls}
\fi

\begin{theorem}\label{thm:fxncall}
Function calls can be implemented reversibly with constant-factor increases in time and space.
\end{theorem}

\iffull
\begin{proof}
Reversible function calls work in a similar manner to normal ones, involving the stack pointer being pushed onto the stack, following the new pointer to the function, and then returning. There are two major changes needed. First, every function that might be logged will create a second function, which is the unrolling version, and appropriate pointers to move to this function when unrolling will be used. Second, we need to be careful to not overwrite anything when moving around. In the case of functions, we know where the function is written in memory, and we know where the function ends. Thus it is easy to specify a \keyw{jump} in our code whose \keyw{comefrom} is the beginning of the function. At the end of the function, we create another \keyw{jump} which returns us to our original place in the code, reading the location which was pushed onto the stack (Figure \ref{fxncall}).
\end{proof} 
\begin{figure}
\begin{pcode}
\keyw{def} \proc{function}($x$): \+
  \mbox{[code1]}
  \keyw{return} $y$ \-
  
\keyw{def} \proc{reverseFunction}($y$): \+
	\mbox{[reverse of code1]}
	\keyw{return} $x$ \-
  
\mbox{[code2]}
$\ul{z}$ += \proc{function}($x$)
\mbox{[code3]}
\keyw{unroll}
\end{pcode}

\begin{pcode}
\keyw{def} \proc{function}($x$): \+
    \keyw{comefrom}(\id{mem}[\id{sp-1}]) \lasttab{//indicates where fxn was called from}
    $\ul{x}$ += \id{mem}[\id{sp-2}] \lasttab{//pulls input from stack}
    \mbox{[code1]} \lasttab{//body of function}
    $\ul{r0} += y$ \lasttab{//special register for return value}
    \mbox{[reverse of code1]}
    $\ul{x}$ -= \id{mem}[\id{sp-2}]
    \keyw{goto}(\id{mem}[\id{sp-1}]) \lasttab{//returns to program} \-
		
	\keyw{def} \proc{reverseFunction}($x$): \+
    \keyw{comefrom}(\id{mem}[\id{sp-1}]) \lasttab{//indicates where fxn was called from}
    $\ul{x}$ += \id{mem}[\id{sp-2}] \lasttab{//pulls input from stack}
    \mbox{[code1]} \lasttab{//body of function}
    $\ul{r0} -= y$ \lasttab{//special register for return value}
    \mbox{[reverse of code1]}
    $\ul{x}$ -= \id{mem}[\id{sp-2}]
    \keyw{goto}(\id{mem}[\id{sp-1}]) \lasttab{//returns to program} \-

\mbox{[code2]}
\keyw{push}($x$)
\keyw{push}(\id{pc})
\keyw{goto}(\mbox{[function start]}) \lasttab{//jump to function}
\keyw{comefrom}(\mbox{[function end]})
\keyw{swap}($\ul{z}$, $r0$) \lasttab{//caller sets $r0$ back to 0}
\keyw{pop}(\id{pc-4})
\keyw{pop}($x$)
\mbox{[code3]} \lasttab{//begin unroll}
\mbox{[reverse of code3]}
\keyw{push}($x$)
\keyw{push}(\id{pc})
\keyw{swap}($\ul{z}$, $r0$) \lasttab{//put output back in $r0$ to be reversed}
\keyw{goto}(\mbox{[reverseFunction start]}) \lasttab{//jump to function}
\keyw{comefrom}(\mbox{[reverseFunction end]})
\keyw{pop}(\id{pc-4})
\keyw{pop}($x$)
\mbox{[reverse of code2]}

\end{pcode}
\caption{High-level and low-level language example of function call and unroll.}
\label{fxncall}
\end{figure}
\fi

\subsection{Memory Management and Garbage Collection}
Allocating and reclaiming free memory are critical tasks to the function of any modern computer. Happily, we show that some basic forms of memory allocation and garbage collection can be done reversibly with only constant factor overhead in time and space.
\iffull
Free lists are a simple type of memory allocation which involves a list of pointers to blocks of free memory. The pointer to the next block of memory is stored in the previous block of memory. Allocation is performed by removing the pointer from the free list and handing it to the program requesting the memory. To deallocate a block of memory, the pointer to that block is put on the end of the free list.  We assume all deallocated memory blocks have been returned to their zeroed out state.
\fi

\begin{theorem}\label{thm:freelist}
Memory allocation using free lists can be done reversibly with constant-factor overheads in time and space.
\end{theorem}

\iffull
\begin{proof}
Free lists can be handled in a similar manner to doubly-linked lists  \ref{sec:LinkedLists}. Passing a pointer from the memory manager to the program or back are reversible operations. To append a block of memory to a list, we simply point the new block to the former head of the list, then point the head pointer to the new block, an operation which is straightforward to reverse.
\end{proof}
\fi

Garbage collection often uses a technique known as reference counting. Reference counting keeps track of the number of references to an object or resource and deallocates the space when it is no longer referenced. In our analysis, we do not charge the cost of freeing or destroying the objects to the algorithm. Since this destruction would need to happen regardless of what sort of garbage collection, if any, was performed we believe the energy costs involved are not a fair representation of the work done by the garbage collector itself.

\begin{theorem}\label{thm:refcount}
Reference Counting can be done reversibly with constant-factor overhead in space and time.
\end{theorem}

\iffull
\begin{proof}
The reference counter requires $\lg k$ bits where $k$ is the number of references to the object. We can augment the creation of a reference to increment the counter, which is a reversible operation. In addition, we can add a simple conditional check on the destruction of a reference to see if the reference counter is at $1$. Since this meets the condition of a protected if, it takes no auxiliary space. We can now decrement the reference, destroy the object, and continue on with removing the reference. Thus, maintaining the counter only incurs a constant factor increase in time and space. The destruction of the object would have occurred regardless, and thus it is not fair to charge its energy cost to the reference counting algorithm.
\end{proof}

Mark and Sweep is another garbage collection algorithm which walks through every object reachable from the heap and marks them. It then goes through all objects in memory and destroys any which are unmarked, and unmarks the rest. Example code can be found below.

\begin{minipage}{0.3\linewidth}
\begin{pcode}
\proc{MarkAndSweep}: \+
\keyw{for each} root variable $r$: \+
    \proc{Mark}$(r)$ \-
\proc{Sweep}$()$ \-
\end{pcode}
\end{minipage}
\begin{minipage}{0.3\linewidth}
\begin{pcode}
\proc{Mark}$ (Object p)$: \+
    \keyw{if} $!p.$marked: \+
        $p.$marked = True;
        \keyw{for each} Object $q$ 
        \enspace referenced by $p'$: \+
            \proc{Mark}$(q)$; \- \- \-
\end{pcode}
\end{minipage}
\begin{minipage}{0.3\linewidth}	
\begin{pcode}		
\proc{Sweep}$()$: \+
    \keyw{for each} Object p 
    \enspace\keyw{in} the heap: \+
        \keyw{if} $p.$marked: \+
            $p.$marked = False \-
        \keyw{else}: \+
            heap.release$(p)$
\end{pcode}
\end{minipage}
\fi

\begin{theorem}\label{thm:marksweep}
Mark and Sweep can be done reversibly with constant-factor overhead in space and time.
\end{theorem}

\iffull
\begin{proof}
The reversible Mark and Sweep will be a slight variation on the standard code above. First, each assignment to $p.$marked can be replaced by (reversibly) toggling that bit because we just checked that it was in the opposite state. At this point the remaining operations are (1) following pointers and control logic, which can be logged and done reversibly with constant-factor overhead, and (2) destroying unmarked objects. Destroying these objects does not impact any of the variables in the control logic and thus can safely stay unlogged.
\end{proof}
\fi

\section{Energy Reduction Techniques}
\label{sec:EnergyReductionTechniques}
This section overviews some general techniques that have been helpful in constructing reversible algorithms and proves some general theorems about algorithms sharing certain properties.

\subsection{Complete Logging}
One very simple, yet surprisingly useful technique is to simply log every step of an algorithm. This incurs a space cost of $O(t(n))$ words where $t(n)$ is the runtime of the algorithm. Although this seems wasteful, the prevalence of linear time algorithms or linear time sub-routines in algorithms makes this important to remember.

\subsection{Reversible Subroutine}
Earlier we saw that function calls can be implemented reversibly. We now give a stronger result for being able to use some reversible subroutines efficiently.

\begin{theorem}
\label{RSubroutineThm}
If we have a fully reversible subroutine whose only effect on the program is through its return value, one need only store the inputs and outputs to this subroutine to later unroll it with only a constant-factor overhead in time.
\end{theorem}
\begin{proof}
To do this, we use a slightly more complicated, two-step unrolling process. First, after the subroutine has initially run, we copy out the output and immediately unroll the subroutine. This copy of the output looks no different to the rest of the program from what would normally be computed, and we've already stipulated that the subroutine cannot alter the program through any other means. When it comes time to unroll the subroutine, we may have lost important logged information needed to take us from the output back to the input. At this point, we run the subroutine forward, recovering all of that needed information. Next we delete the copy of the output and unroll the subroutine normally. 
\end{proof}
\iffull
\subsection{Pointer Swapping}
Depending on how it is managed, pointer destruction will generally take energy. Many data structures liberally create and destroy pointers to various nodes or other pieces of data. In many cases, one can keep back-pointers or self-referential pointers which would allow them to simply be moved around. In general, moving pointers around has zero energy cost, and thus this is one simple way to reduce the energy complexity of many data structures.
\fi
\subsection{Data Structure Rebuilding}
When attempting to implement data structures which support insert and delete operations reversibly, we run into a new challenge. Often the insertion or deletion operation will create some amount of garbage data which is necessary to reverse it in the future. We also need the result of the operation to remain in place, so we cannot immediately reverse the operation. Thus over the life of the data structure, its size will depend on the total number of insert and delete operations, rather than just the number of elements in the data structure. To circumvent this we can use a technique we call \emph{periodic unwinding}. Note, this technique depends exceedingly on what is considered the data-structure and what is an algorithm that uses the data-structure. This is discussed more below and in Section~\ref{sec:DataStructures}.

\begin{theorem} \label{thm:periodic-unwinding}
If a data structure which allows reversible insertions, deletions, and traversals can be constructed reversibly from $k$ insertions in $O(k)$ time and space, and its operations can be performed reversibly with constant-factor overhead in time and space, then it can be maintained reversibly in amortized time with constant-factor overheads in space and time via \emph{periodic unwinding}.
\end{theorem}
\begin{proof}
If there are only $O(n)$ deletions, then we can simply log and unroll all of the operations. If not, we need to keep track of the number of insertions and deletions that have occurred because the last rebuilding. We will keep track of these counts and increment them as part of the insertion or deletion routine. We also track the number of elements in the data structure. Whenever a delete is called, we then check to see whether the number of deletions is more than twice the number of nodes in the data structure. If this is true, we will proceeded to rebuild the data structure. We perform a reversible traversal of the tree, with the addition that we make an extra copy of the inserted data at every node. Now that we have this copy, we can proceed to unroll the data structure, clearing the log. Once this is done we construct a new data structure with the same values as before the reversing, but with none of the accumulated garbage. Construction of the new data structure takes $O(n)$ time. To trigger a rebuilding, we must have called delete a larger number of times than the size of the data structure we are building. We charge the amortized constant cost per element being added to the new data structure to the number of deletes performed, giving us constant amortized time. Our counters and copies of the data all require $O(n)$ space, meaning we never use more than a constant-factor overhead in space. 
\end{proof}

With this method, after rebuilding and clearing the log, the data structure can no longer provide any information about past items which were deleted from the data structure. This is covered by the assumption that the algorithm interacting with the data structure makes copies of all of its inputs and if it needs them to reverse itself, it is responsible for maintaining that information. Thus, depending on how the data structure is being used, this technique can be superfluous or very powerful.

\section{Data Structures}
\label{sec:DataStructures}
Data structures are meant to be used in the context of an algorithm. In the standard model for algorithms, we can draw a nice abstraction between these two and analyze their properties separately. We also wish to do so in the energy complexity model; however, we need to be more careful about the responsibilities of the data structure and the algorithm for maintaining information and reversibility. First, we assume the data structure is only accessed through the prescribed operations; we don't want the algorithm irreversibly altering stored elements or manually altering the data structure in an unknown way. Second, if it is a reversible data structure, every operation must have a reverse operation. Third, when inserting, the algorithm gives a copy of the data to the data structure and is responsible for maintaining its own reversibility after an insert has been reversed. Fourth, the algorithm will handle zeroing of the bits of elements removed from the data structure. For this purpose, the common delete operation will be replaced with \emph{Extract$(x)$} which removes an element from a data structure and returns it; however, we will generally still call this operation delete. Fifth, for a reversible algorithm, it is responsible for making the correct calls to reverse functions to reset the data structure.

This is certainly not the only way to treat this interface. We could just as well require the data structure to remember all of the calls performed on it so it could reverse itself upon command. Similarly, we could imagine that a deleted item cannot be handed back to the algorithm, but must in some way be removed by the data structure, most likely when unrolling. We've chosen our conventions because it more closely matches our idea of how subroutines should work and because it is clearer to us how to analyze such cases.

\subsection{Stacks and Linked Lists}
\label{sec:LinkedLists}
\begin{theorem}
Doubly-linked lists can be implemented reversibly with constant-factor overheads in time and space.
\end{theorem}
\iffull
\begin{proof}
Doubly linked lists need to support adding, removing, and finding items. To search for an item $x$, one begins at one end of the list, performs a comparison of $x$ and the next value in the data structure, and then either returns a match or follows a pointer. Match-checking and following pointers can be done reversibly. In this case, we change neither $x$ nor the values in the data structure, so we can use a Protected If to perform our comparisons. This leaves us with a constant factor overhead in time and space.

To add an element in position $i$, we follow pointers in the list until we reach the element at position $i-1$. Our counter can be kept reversibly as was done for For Loops \iffull\ref{sec:forloops}\fi \ifabstract\ref{sec:control-logic}\fi. We move the forward pointer from element $i-1$ into our new element, and then similarly move the back-pointer from element $i$ into our new element. Next we copy the pointer to our new element into the forward pointer slot of element $i-1$ and into the back-pointer slot of element $i$. We then destroy our copy of the pointer to the new element. Since moves, copies, and destroy copies are all reversible, we incur only constant factor overheads in insertion of a new item.

Removing an element from a data-structure is very similar to adding one, and also only requires moves, copies, and delete-copies. To remove element $i$, we delete the copies of it's pointer which elements $i-1$ and $i+1$ contain. We then copy the forward and back pointers from elements $i-1$ and $i+1$ into the appropriate positions. The pointer to our removed element is then handed back to the algorithm which called the remove item method on our data structure.

\end{proof}
\fi

\begin{corollary}
Stacks, queues, and dequeues can be implemented reversibly with constant-factor overheads in time and space.
\end{corollary}

\subsection{Dynamic Arrays}
\iffull
Dynamic arrays store elements in an ordered array that grows and shrinks depending on the number of elements currently being stored. They support the basic operations \proc{add}$(e)$, \proc{delete}$()$, and \proc{query}$(i)$. \proc{add} and \proc{delete} append or remove an element from the end of the array respectively. \proc{query} returns the element stored at a specific index of the array. Dynamic arrays can also support the optional operations \proc{add}$(i, e)$ and \proc{delete}$(i)$ which specify at which index to add or delete an element, adjusting all elements stored after that index accordingly.

A dynamic array stores a $\id{length}$ attribute and a $\id{size}$ attribute. The $\id{length}$ attribute represents the number of active elements are in the array. The $\id{size}$ attribute represents the amount of space allocated for the array. $\id{length}$ is incremented and decremented corresponding to \proc{add} and \proc{delete} operations. Upon an \proc{add}, if $\id{length}=\id{size}$, we double the array size by allocating new memory, copying over the elements and deallocating the old array. Similarly upon a \proc{delete}, if $\id{length}=\id{size}/4$, we halve the array size by deallocating the second half. 
\fi
\begin{theorem} \label{thm:dynamic-array}
Dynamic arrays can be implemented reversibly with constant time and space overhead with an extra bit of space for \proc{add}/\proc{delete} operations. Size of the structure grows with the number of \proc{add}/\proc{delete} operations.
\end{theorem}

\begin{proof}
We now consider how to handle these operations reversibly. Because both \proc{add}/\proc{delete} operations work at the end of the array, we check the $\id{length}$ attribute to find where to perform the reverse operation. For \proc{reverse-add}, we remove the element and decrement $\id{length}$. And for \proc{reverse-delete}, we must log the deleted element to add it back and increment $\id{length}$. We must also consider how to handle table doubling. On an \proc{add}/\proc{delete} operation, table doubling (or halving) occurs based on the result of a single comparison of $\id{length}$ and $\id{size}$. We can log a single bit representing the result of this comparison for each \proc{add}/\proc{delete} operation that will indicate whether a table doubling (or halving) needs to be reversed.

This bit is necessary in order to undo table doubling because we can not determine whether a table doubling operation occurred just by looking at the resulting $\id{length}$ and $\id{size}$ attributes. For example, consider a table with length $n$ and size $2n$ where the last operation was an \proc{add}. This state could have been reached in two ways. (1) length $n-1$, size $n$ incurring a table doubling; (2) length $n-1$, size $2n$.

We maintain the dynamic array to preserve the order and length of its elements in the reverse direction, thus a \proc{reverse-query} operation can be run in the same way as a \proc{query} operation by simply making the same query again. Periodic rebuilding of this data structure follows from Theorem \ref{thm:periodic-unwinding} because all operations are reversible with constant factor overhead and rebuild can be done in linear time.
\end{proof}

\subsection{AVL Trees} \label{sec:avl-trees}
Using and maintaining standard AVL trees incurs an energy cost proportional to the time of the associated operations.

\begin{theorem} \label{thm:avl-trees-search}
\proc{Search}$(x)$ can be performed on standard AVL trees in $\Theta(\lg n)$ time, $O(1)$ auxiliary space and $O(\lg n)$ energy.
\end{theorem}
\iffull
\begin{proof}
To search an AVL tree we follow a pointer to the root of the tree and compare our value to the value in the current node. If it is the same we have found the value. If it is smaller we follow the left child pointer in the tree and repeat the comparison, if it is larger we follow the right child pointer. We must check that those pointers are not null, and if they are we terminate knowing the item is not in the tree. Following pointers incurs no energy cost, but each irreversible compassion takes one bit of energy. AVL trees are guaranteed to have a max depth of $O(\lg n)$ and thus we use $O(\lg n)$ time and energy. Since we need only maintain the pointer to the current node and the value we are searching for there is only a constant amount of space needed.
\end{proof}
\fi

\begin{theorem} \label{thm:avl-trees-insert}
\proc{Insert}$(x)$ can be performed on standard AVL trees in $\Theta(\lg n)$ time, $O(1)$ auxiliary space and $O(w\lg n)$ energy.
\end{theorem}
\iffull
\begin{proof}
To insert a new element, we perform comparisons and follow child pointers, as with a search, until we reach a null pointer. We replace this with our item to be inserted, potentially deleting the value previously in the memory address holding the pointer. The tree may now need to be rebalanced. Every insertion may incur up to $O(\lg n)$ rotations. The details of the four different cases for rotation can be seen in \cite{CLRS} and the time analysis remains the same. For energy, we note that all rotations in the standard AVL algorithm involve a constant number of comparisons, each costing $1$ unit of energy, and the creation and deletion of at least one pointer, each costing $O(w)$ energy. This leads to a total of $O(w\lg n)$ energy. If we are slightly more clever and use pointer swapping techniques, the energy drops to a constant amount per rotation and comparison, resulting in $O(\lg n)$ energy. 
\end{proof}
\fi

\begin{theorem} \label{thm:avl-trees-delete}
\proc{Delete}$(x)$ can be performed on standard AVL trees in $\Theta(\lg n)$ time, $O(1)$ auxiliary space and $O(w\lg n)$ energy.
\end{theorem}
\iffull
\begin{proof}
To delete an item we first perform a search on the tree for the element to be deleted, incurring the associated costs. The element and the pointer to it are then deleted at a cost of $w$ units of energy. We may then need to rebalanced the tree. Just as with insertion, re-balancing after a deletion may incur up to $O(\lg n)$ rotations. The standard AVL algorithm involves the creation and deletion of a pointer during a rotation, each costing $O(w)$ energy. If we are slightly more clever and use pointer swapping techniques, the energy drops to a constant amount per rotation and comparison, resulting in $O(\lg n)$. 
\end{proof}
\fi

We will show that, provided only \proc{Search} and \proc{Insert} operations are invoked, reversible AVL trees can be maintained with only constant-factor auxiliary space consumption. If \proc{Delete} is to be invoked, then the structure will accumulate an extra $\Theta(k)$ words of space for $k$ \proc{Delete} operations invoked over the lifetime of the tree. Such space consumption can still be made reasonable within the context of a larger algorithm, provided that runs of \proc{Insert} and \proc{Delete} form a small part of the algorithm, and are unwound periodically to refresh log space.

Note that because these algorithms employ only conditional branches which do not modify their conditions (for example, in \proc{Search} when comparing a value against a node of the tree to choose a branch, we leave the value of the comparison intact post-search), they are completely reversible with no logging penalty.

\begin{theorem} \label{thm:avl-trees-search-reversible}
\proc{Search}$(x)$ can be performed on reversible AVL trees in $\Theta(\lg n)$ time, $O(1)$ auxiliary space and 0 energy.
\end{theorem}
\begin{proof}
Provided that our reversible AVL tree is constructed using two-way nodes, performing \proc{Search}$(x)$ reversibly is straightforward. Upon reaching a node $v$ in the tree, we compare its value with $x$ and use the resulting bit to determine whether to jump left or right. After jumping, as we have maintained in memory a pointer back to $v$, we can compare $x$ to $v$ again to destroy this bit and free the space gain. Once the final node is reached and our answer is found or determined to be absent, we can log our result somewhere and reverse our computation to destroy any remaining garbage bits. This procedure uses constant auxiliary space, and produces our answer reversibly in $O(\lg n)$ time.
\end{proof}

\proc{Insert} is the next operation we address. It includes the task of reversibly rebalancing the tree, a slightly more complicated task than that of \proc{Search}.

From a given tree, there may be multiple legal trees that underwent a different rotation to produce it. Thus, if we didn't store any auxiliary information about the AVL rotations as we performed them, it would not be possible to immediately reverse a tree's configuration. A key insight into the space consumption of this process is to note that, for a tree containing $n$ unique elements, each of those elements must occupy at least $\Omega(\lg n)$ bits of space each on average (in the word model in particular, each element takes a constant $w = \Theta(\lg n)$ bits of space). Thus, we must store a $\Theta(\lg n)$-sized entry to a rotation log for each inserted element. The space cost of this logging is absorbed into the space cost of the element's value in the tree itself. This is the premise of the following theorem:

\begin{theorem} \label{thm:avl-trees-insert-reversible}
\proc{Insert}$(x)$ for $x$ not yet present in the reversible AVL tree can be performed in $\Theta(\lg n)$ time, preserving the $\Theta(n)$ space cost of the tree and using 0 energy.
\end{theorem}
\begin{proof}
Insertion consists of traversing the reversible AVL tree, adding the new element to the tree, and making any rotations that are necessary to balance the tree.

Traversing the tree can be done reversibly as in \proc{Search}$(x)$, and we refer to its proof for reversibility. Once we know where $x$ is to be added into the tree, we create a new node for it and proceed to rebalance the tree.

During the balancing step, rotations begin at the lowest level and proceed upward. To perform our operations reversibly, we will keep a log of every rotation performed at each of the $\lg n$ levels of the tree. For each level, we will store 01 for a right-rotation, 10 for a left-rotation, and 00 for no rotation. By keeping this log, we have enough information to go in the reverse direction, proceeding from the top of the tree to its bottom and checking $x$'s value against those of the encountered nodes to progress. In this way, we keep our \proc{Insert}$(x)$ action reversible.

Each log entry need only store a number of bits proportional in size to the maximum height of the tree. Because there are $n$ unique entries in the tree at any given time, each call to \proc{Insert} incurs only $O(\lg n)$ bits of space cost. Thus our log, which stores $\Theta(\lg n)$ additional bits per element, results in only a constant-factor increase in the space consumption of the tree. This holds even if deletions from the tree have also been performed, as inserting into a tree will always grow the tree's space consumption asymptotically by $\Omega(\lg n)$ bits per insertion, while the log size will grow by $O(\lg n)$ bits per insertion. The space consumption of the tree is thus preserved to within constant factors.
\end{proof}

\begin{theorem} \label{thm:avl-trees-delete-reversible}
\proc{Delete}$(x)$ can be incorporated into reversible AVL trees, taking $O(\lg n)$ time and incurring an additional $\Theta(k \lg n) = O(kw)$ bits or $O(k)$ words of space for $k$ delete operations.
\end{theorem}
\iffull
\begin{proof}
\proc{Delete}$(x)$ simply requires the same traversal operations as \proc{Search}$(x)$, followed by the extraction of the desired node. Given that it takes no more than $O(\lg n)$ bits to specify a location in the tree, and because we will need $\Theta(\lg n)$ bits to store the subsequent rotation information, we must, as is the case for \proc{Insert}$(x)$, store additional information in the log proportional to the size of the original elements. This incurs a cost while simultaneously reducing the size of the tree, and thus for $k$ deletions we require $\Theta(k \lg n) = O(kw)$ extra bits or $O(k)$ words of space beyond that occupied by the tree itself, for a total space consumption of $O(n + k)$ words.
\end{proof}
\fi

\subsection{Binary Heaps}
Binary heaps are binary trees with the property that all children are either less than their parent, in the case of max-heaps or greater in the case of min-heaps. They support inserting an arbitrary element into the data-structure as well as popping the max/min element. We give a brief energy analysis as well as an efficient reversible algorithm.
\iffull
\begin{figure}
\small
\begin{minipage}{0.4\linewidth}
\begin{pcode}
\proc{MaxHeapify}$(A, i)$: \+
  ($i$, \id{left}) = $(i, 2*i)$
  ($i$, \id{right}) = $(i, 2*i + 1)$
  ($i$, \id{largest}) = $(i, i)$
  \keyw{if} \id{left} $\leq$  $A$.\id{length}() 
  \enskip   \keyw{and} $A[\id{left}] > A[\id{largest}]$: \+
    \id{largest} = \id{left} \-
  \keyw{if }\id{right} $\leq$ $A$.\id{length}()
  \enskip   \keyw{and} $A[\id{right}] > A[\id{largest}]$: \+
    \id{largest} = \id{right} \-
  \keyw{if} \id{largest} $\neq i$: \+
    Swap$(A[\id{largest}], A[i])$
    \proc{MaxHeapify}$(A, \id{largest}) $\- \-
		
\proc{Insert}$(a)$: \+
  \id{heapSize} = \id{heapSize}$ + 1$
  $A[\id{heapSize}] = a$
  \keyw{for} $i = 0$ to $\lfloor \lg$(\id{heapSize})$ \rfloor$: \+
    \keyw{if} $A[\lfloor$\id{heapSize}$/2^{i+1}\rfloor$]
      \enskip $< A[\lfloor$\id{\id{heapSize}}$/2^i\rfloor]$: \+
      swap$(A[\lfloor$ \id{heapSize}$/2^{i+1}\rfloor]$, 
               $A[\lfloor$ \id{heapSize}$/2^i\rfloor ])$ \- \- \-
	
\proc{DeleteMax}$()$: \+
  Swap$(A[1], A[\id{heapSize}])$
  $(A[\id{heapSize}], \id{maxValue}) = (0, A[\id{heapSize}])$
  \proc{MaxHeapify}$(A, 1)$
  \keyw{Return} \id{maxValue}
\end{pcode}

\caption{Binary Heap Code}
\label{Binary heap code}
\end{minipage}
\hfil\hfil
\begin{minipage}{0.4\linewidth}
\begin{pcode}
$\proc{ReversibleMaxHeapify}(A, i)$: \+
  $(i, \ul{\id{left}}) = (i, 2*i)$
  $(i, \ul{\id{right}}) = (i, 2*i + 1)$
  \keyw{if} $\id{right} \leq A.$\id{length}
     \keyw{and} $A[\id{right}] > A[\id{largest}]$: \+
    $\ul{\id{largest}} = \id{right}$ \-
  \keyw{else if} $\id{left} \leq  len(A)$ 
        \keyw{and} $A[\id{left}] > A[\id{largest}]:$ \+
    $\ul{\id{largest}} = \id{left}$ \-
  \keyw{if} $\id{largest} \neq 0$: \+
    $swap(A[i], A[\id{largest}])$
    $\proc{ReversibleMaxHeapify}(A, \id{largest})$ \-
  $\id{left} = \id{left} - 2*i$
  $\id{right} = \id{right} -2*i - 1$ \-

\proc{ReversibleDeleteMax}$()$: \+
  $\id{deleteCount} = \id{deleteCount} + 1$
  \keyw{if} $\id{deleteCount} > \id{heapSize}:$ \+
    $\proc{Rebuild}(A)$ \-
  $\left(A[1], A[\id{heapSize}]\right) = \left(A[\id{heapSize}], A[1]\right)$
  $\left(A[\id{heapSize}], \ul{MaxValue}\right) = \left(0, A[\id{heapSize}]\right)$
  $\proc{ReversibleMaxHeapify}(A, 1)$
  $\keyw{Return} \id{MaxValue}$

\end{pcode}
\caption{Reversible Binary Heap}
\label{Fig:ReversibleBinaryHeapCode}
\end{minipage}
\end{figure}
\fi
\begin{theorem} \label{thm:binary-heap-insert}
Binary Heaps can have items inserted irreversibly with $\Theta(\lg n)$ time, $\Theta(1)$ space, and $\Theta(\lg n)$ energy; or reversibly with $\Theta(\lg n)$ time, $\Theta(1)$ space, and $0$ energy.
\end{theorem}

\iffull
\begin{proof}
Insert begins by incrementing the heap size and assigning an empty entry in the array to the value to be inserted. These are both reversible operations. We then proceed to fix the heap by checking whether the current node is smaller than the parent, and if so swapping them. Normally the comparison and loop would destroy bits, leading to a $O(\lg n)$ energy cost, however this is a protected if and simple for loop so Lemmas~\ref{thm:if} and \ref{thm:simplefor} allow us to do so reversibly with no energy cost.  
\end{proof}
\fi

\begin{theorem} \label{thm:binary-heap-delete}
Binary Heaps can have the root node deleted irreversibly with $\Theta(\lg n)$ time, $\Theta(\lg n)$ space, and $\Theta(w \lg n)$ energy.
\end{theorem}

\iffull
\begin{proof}
DeleteMax moves out the largest element of the heap and then restructures it with a recursive call to MaxHeapify. This can be called up to $O(\lg n)$ times. Every call to MaxHeapify looks at the children of the current node, determines if it is larger than the value in the current node, and if so swaps them. The maintenance of the current largest value as well as those of the children involves overwriting those values. The conditionals additionally take a constant amount of energy, leading to an energy complexity of $O(w\lg n)$. The standard time and space analysis remains unchanged.
\end{proof}

To implement Reversible Binary Heaps, we make two changes:
MaxHeapify is written reversibly,
and we use periodic unwinding to keep $O(n)$ space. The Rebuild function is not given here, but can be derived from Theorem~\ref{thm:periodic-unwinding} 
\fi

\begin{theorem} \label{thm:binary-heap-delete-reversible}
Binary Heaps can have the root node deleted reversibly with $\Theta(\lg n)$ time, $\Theta(\lg n)$ space, and $0$ energy.
\end{theorem}

\iffull
\begin{proof}
MaxHeapify has been rewritten to ensure that no values are ever overwritten during the function call by reordering the logic and subtracting away known values from variables. Figure~\ref{Fig:ReversibleBinaryHeapCode} shows the substitutions to ensure no values are overwritten. We then note that the conditionals are protected if statements, and thus reversible by Theorem~\ref{thm:if}.
\end{proof}
\fi

As with AVL trees, binary heaps subject to $k$ insertions and deletions will accumulate an extra $O(k)$ space to be maintained. In some cases this can be resolved by periodic unwinding.

\section{Algorithms}
\label{sec:Algorithms}

This section includes the analysis for the time, space, and energy complexity of several standard algorithms in our model. We also give a number of improved algorithms. Some of our results for algorithms with zero energy complexity are similar to results claimed or proved in \cite{Frank99} about reversible algorithms. However, we prove these results within our own model, which differers slightly from \cite{Frank99}.

\subsection{Sorting}
Sorting is among most fundamental and well understood algorithmic problems. In this section we give reversible algorithms for comparison and counting sorts which match the time and space complexities of know irreversible algorithms. It is especially interesting to see this is achievable despite the known entropy change during comparison sorts which give us a lower bound on their time complexity.

\subsubsection{Comparison Sort} \label{sec:comparison-sort}
\begin{theorem}
A comparison sort destroying its input must consume $\Omega(\lg n!)$ energy.
\end{theorem}

\iffull
\begin{proof}
Given a list of $n$ unique items, there are $n!$ permutations of that list, only one of which is in sorted order. Thus the end-to-end entropy of a sorting function is $\lg (n!)$, resulting in an energy complexity of $\Theta(\lg n!)$. This was noted in \cite{Frank99}.
\end{proof}
\fi

We can achieve this energy bound with Merge Sort. Merge Sort takes in an array of numbers, recursively calls itself on half of the array until it reaches the sorted array of size 1. It then merges the returned arrays by iteratively comparing the smallest values in each array and moving it to the beginning of a new sorted array.

\begin{theorem} \label{thm:comparison-sort}
Comparison sort destroying its input can be done in $\Theta(n \lg n)$ time,  $\Theta(n)$ space, and  $\Theta(n \lg n)$ energy.
\end{theorem}

\iffull
\begin{proof}
Here we look at the standard merge sort which runs in $\Theta(n \lg n)$ time and  $\Theta(n)$ space. This sorting algorithm recursively divides up the items to be sorted, then sorts increasingly larger groups during the merge step. For the merge step, one takes two sorted lists and starts comparing the smallest element. The smaller one is put into the new sorted list and the process continues. The $\Theta(n \lg n)$ comparisons will consume $\Theta(n \lg n)$ energy. We also need to ensure that all of the entries are moved into empty arrays, so no overwriting of values occurs. If one copies and deletes during this process there will be an extra factor of $w$ in the energy cost.
\end{proof}
\fi

If we do not destroy the input and are careful with our algorithm, we can do better:

\begin{theorem} \label{thm:comparison-sort-reversible-word}
Comparison sort, not destroying its input, if performed on an array of $n$ w-bit elements with $\lg n = O(w)$, can be done reversibly in $\Theta(n \lg n)$ time, $\Theta(n)$ auxiliary space, and 0 energy.
\end{theorem}

\begin{proof}
This algorithm is a modification of Merge Sort. In summary, we will we augment each element of the array with its index in the array, and so-equipped shall reversibly merge sort the elements in $\Theta(n \lg n)$ time. After we remove these indices from the sorted array, the output of the algorithm will be the original array $L$ and a sorted copy of the array $L_{sorted}$.

During a traditional Merge Sort, there are three main steps: dividing the array in two, recursing on each half of the array, and merging the two resultant arrays into a complete, sorted array.

The first step, dividing an array in two, is reversed trivially: Given the two resultant lists of such an operation $\lrtos$ and $\lstot$, the original subarray $\lrtot$ is their concatenation. 

The second step, recursing, will be reversible if our entire algorithm is reversible. We know the base case of sorting a size-1 array is reversible, and thus if steps 1 and 3 of our algorithm are also reversible, then this recursive step will be as well.

The third step, in contrast to the first two, presents us with some difficulties. Given a resultant, fully-merged subarray $L_{sorted}[r:t]$, it is not at all obvious how to go backwards, i.e. how to reproduce the input subarrays $L_{sorted}[r:s]$ and $L_{sorted}[s+1:t]$. Information will be lost in this merge step, and to allow our algorithm to be done reversibly, we must find some way to preserve it.

Our augmentation makes this step possible. Before sorting, we transform $L$ into a new array $L'$ of twice the size, which consists of the elements of $L$ each augmented with their index in $L$. Thus, the elements of our array are 2-tuples $(v, i)$ of each element's value and original location in $L$. The above transformation is sufficient to make the merge step reversible. To see why, consider any step in the algorithm in which we are trying to merge two sorted subarrays $\lrtos$ and $\lstot$. Denote the merge subroutine that we are trying to compute as $M_{s+\frac{1}{2}}(\lrtos, \lstot)$; that is, we are merging around a pivot $s+\frac{1}{2}$, determined by which step of the algorithm we are presently carrying out. All elements $(v,i)$ with $i < s+\frac{1}{2}$ must have come from $\lrtos$, and elements with $i > s+\frac{1}{2}$ from $\lstot$. Given a resultant array $L_{sorted}[r:t] = M_{s+\frac{1}{2}}(L_{sorted}[r:s], L_{sorted}[s+1:t])$, we can reverse the merge operation step-by-step simply by checking each element's index against the pivot $s+\frac{1}{2}$ to determine where it came from. This enables us to construct two-way branches that perform the merge in a way that is instantaneously reversible. Because the pivot is fixed for each step in the algorithm, no information is lost in computing and decomputing it, and thus this step of the algorithm may be implemented reversibly with only constant additional auxiliary space.

The output of the above algorithm is a list $L_{sorted}'$ of $(v,i)$ tuples sorted in the $v$ keys (whereas our original list $L'$ was ``sorted'' in the $i$ keys).
Now we need only remove the auxiliary indices from the elements $(v,i)$ to produce our unaugmented sorted list $L_{sorted}$. This step must be handled with care to ensure that every step is reversible. We begin by reproducing the original array $L$ via a single pass over $L_{sorted}'$, a simple operation that has not yet destroyed any data. Next, we copy out $L_{sorted}$, the final, sorted array that we care about. What remains is to dismantle $L_{sorted}'$, and here we shall employ a special trick: We will perform a single pass over the \emph{original} array $L$, and for every value $v$ encountered we will perform a logged binary search for this element in the \emph{unaugmented} sorted list $L_{sorted}$. When the element is found, we will know the value $v$, index $i$, and the location of the element $(v,i)$ in the augmented sorted array $L'_{sorted}$. This is sufficient to destroy this element, setting its entry to zero before unrolling the log of our binary search. The complete dismantling operation uses only $\Theta(\lg n)$ additional logging space total, and only takes time $\Theta(n \lg n)$, so our runtime and space consumption are preserved.

This algorithm is instantaneously reversible at every step, and could be implemented using only simple for loops and two-way conditional branches. Thus, the algorithm is completely reversible under our model. Given an array of size $n = O(2^w)$ which occupies $nw$ space in memory, we can reversibly comparison-sort the array using $\Theta(nw)$ bits of auxiliary space in $\Theta(n\lg n)$ time, matching the best irreversible algorithm to within constant-factors of space and time.
\end{proof}

\begin{theorem} \label{thm:comparison-sort-reversible-general}
Comparison sort, not destroying its input, can be done reversibly on an array of $n$ $d$-bit elements which require $nd$ space in $\Theta(n \lg n)$ time, $\Theta(n d)$ auxiliary space, and 0 energy.
\end{theorem}

\begin{proof}
As we saw in the preceding theorem, reversible comparison sort is straightforward to perform if we first augment each element in the array with a number corresponding to its index in the original list. When the size of the values $d$ is $\Omega(\lg n)$, then  we attain the optimum space bound as the $\lg n$-sized indices get absorbed into the space cost of the $d$-bit elements. However, we are faced with a conundrum when $d$ is $o(\lg n)$.

To handle this case, we shall employ counting sort in order to reduce the problem of sorting $L$ to the problem of sorting the \emph{unique keys} of $L$. We utilize a reversible AVL tree (described in Section~\ref{sec:avl-trees}) to achieve this.

This algorithm works by reducing the array $L$ only to its unique elements, to sort those elements, and finally to perform a Counting Sort of the original array, consulting our sorted elements to determine the final order. Let $k$ be the number of distinct elements of $L$. We employ a reversible AVL tree, with actions carefully specified so as to keep them reversible. First, we read the distinct elements of $L$ into the tree, bringing it to a size $O(kd)$, and keeping an $O(n)$-bit uniqueness log and an $O(nd)$-bit rotation log (see reversible AVL trees discussion) as we go. This step takes $O(n \lg n)$ time. Next, we apply Counting Sort on the original array (see Section~\ref{sec:counting-sort}), consulting the static tree in $O(\lg n)$ time for each element and achieving the $O(n \lg n)$ runtime in this step as well. The output array may be copied and the entire algorithm reversed (Note: not an unrolling of a log, but rather an execution of a reversed version of the algorithm) to leave us with our desired arrays $L$ and $L_{sorted}$.

This algorithm works by reducing the array $L$ only to its unique elements, to sort those elements, and finally to perform a Counting Sort of the original array, consulting our sorted elements to determine the final order. Let $k$ be the number of distinct elements of $L$. We employ a reversible AVL tree, with actions carefully specified so as to keep them reversible. First, we read the distinct elements of $L$ into the tree, bringing it to a size $O(kd)$, and keeping an $O(n)$-bit uniqueness log and an $O(nd)$-bit rotation log (see reversible AVL trees discussion) as we go. This step takes $O(n \lg n)$ time. Next, we apply Counting Sort on the original array (see Section~\ref{sec:counting-sort}), consulting the static tree in $O(\lg n)$ time for each element and achieving the $O(n \lg n)$ runtime in this step as well. The output array may be copied and the entire algorithm reversed (Note: not an unrolling of a log, but rather an execution of a reversed version of the algorithm) to leave us with our desired arrays $L$ and $L_{sorted}$.

In terms of the intricate details glossed over above, the most involved are in the first step: reversibly constructing a reversible AVL tree out of the unique elements of $L$. We proceed as follows: making a single pass over our array $L$, we add every element into the tree. If an element is the first of its exact value to be encountered, we store a corresponding uniqueness bit as true and add the element to the tree. If an element's value already exists in our tree, we store its uniqueness bit as false and move on (never adding duplicates to the tree). These $n$ uniqueness bits allow us to reverse the algorithm, as we know for which elements we modified the tree and on which ones we did not. By theorem \ref{thm:avl-trees-insert-reversible}, these insertions take $\Theta(nd)$ space and $O(n \lg n)$ time. 

Once the AVL tree is constructed, the rest of the algorithm is a straightforward Counting Sort with a slower $\Theta(\lg k)$ lookup time, yielding the $O(n \lg k)$ time which is optimal for Comparison-based Sort. In addition to the fully-reversible AVL tree data structures, our algorithm employs only simple for-loop passes over the input and reversible two-way branching in the AVL tree, ensuring its reversibility. The entire algorithm takes $\Theta(nd)$ space and $O(n \lg n)$ time, matching the best irreversible comparison sorts up to constant-factors of space and time.
\end{proof}

\iffull
We can also create reversible version of other sorting algorithms.

\begin{pcode}
$\proc{Duplicated-Insertion-Sort}(A)$ \+  
    $\id{for} i=1$ \id{to} $A.$Length: \+
        $ \ul{\keyw{count}} = 0 $
        $ \id{for} j=1$ \id{to} $A.$Length: \+
            $ \id{if} A[i] > A [j]$: \+
                $ \keyw{count}$ += $1$  \-
        $B[\keyw{count}] = A[i] $
\end{pcode}
\fi

\begin{theorem} \label{thm:insertion-sort-reversible}
Reversible Duplicated Insertion Sort runs in $\Theta(n^2)$ time,  $\Theta(n)$ space, and  $0$ energy.
\end{theorem}
\iffull
\begin{proof}
The primary adaption we will make to the algorithm is unrolling after each iteration of the outer for loop. The for loops are simple for loops and thus incur only constant-factor overheads when done reversibly. During each inner loop, we perform $n$ comparisons, $O(n)$ increments, and a single addition, because we know the output array, $B$, will be empty. This takes $\Theta(n)$ space in the log and after every iteration we will unroll all operations except for the addition to $B$. This unrolling returns our counter back to zero. The counter and log are empty after each unrolling, so we are finished once we perform the final insertion into $B$.
\end{proof}
\fi
\subsubsection{Counting Sort} \label{sec:counting-sort}
Counting sort involves counting the number of elements at or below a specific value, and then running through them and adding them to an array based on how many elements are below them. This achieves $\Theta(n+k)$ time and space where $k$ is the size of the maximum integer to be sorted.
\iffull
Pseudocode is given below.
\begin{pcode}
$\proc{Counting-Sort}(A,k):$ \+                                                                
     \keyw{for} $j$=1 \keyw{to} \id{length}[A]:       //simple for loop, $\lg(\id{length}[A])$ energy \+
          $C[A[j]]$=$C[A[j]]$+1        //inc is free \-
     \keyw{for} $i$=1 \keyw{to} $k$:                 //simple for loop, $\lg k$ energy \+
          $C[i]$=$C[i]$+$C[i-1]$       //addition is free \-
     \keyw{for} $j$=\id{length}$[A]$ \id{to} 1:        //simple for loop, $\lg(\id{length}[A])$ energy \+
          $B[C[A[j]]]$=$A[j]$       //free if we assume the sorted elements are unique
                                // otherwise unconstrained replacement costs w energy
          $C[A[j]]$=$C[A[j]]$-1    //dec is free \-
     \keyw{delete}($\ul{C}$)         //free delete because we promised C is now empty
     \keyw{return}($B$) 

log \proc{Counting-Sort}$(A,k):$ \+                                                                
     \keyw{for} $j$=1 \keyw{to} \id{length}$[A]$:   //simple for loop, $\lg(\id{length}[A])$ space \+
          $C[A[j]]$=$C[A[j]]+1$        //inc is free \-
     \keyw{for} $i$=1 \keyw{to} $k$:             //simple for loop, $\lg k$ space \+
          $C[i]$=$C[i]$+$C[i-1]$   //reversible and the answer is always available, \-
     \keyw{for} $j$=\id{length}$[A]$ \id{to} 1:     //simple for loop, $\lg(\id{length}[A])$ space \+
          $B[C[A[j]]]$=$A[j]$     //free if we assume the sorted elements are unique
                             //otherwise unconstrained replacement costs w space
          $C[A[j]]$=$C[A[j]]$-1    //dec is free \-
     \keyw{delete}($\ul{C}$)      //free delete because we promised C is now empty
     \keyw{return}($B$)
\end{pcode}
\fi

\begin{theorem} \label{thm:counting-sort}
Counting Sort can be done in $\Theta(n + k)$ time,  $\Theta(n + k)$ space, and  $\Theta(wn + \lg k)$ energy.
\end{theorem}

\iffull
\begin{proof}
The time and space complexity comes from maintaining and iterating over the input as well as an array of size $k$ which represents the range of the numbers. The unconstrained replacement when sorting the numbers results in an $O(wn)$ energy cost and irreversibly running the for loop over the range requires $\Theta(\lg k)$ energy, because we lose the index which ranged up to $k$.
\end{proof}
\fi

\begin{theorem}
If all entries are unique, then Counting Sort has an energy complexity of $\Theta(\lg n + \lg k)$ energy.
\end{theorem}

\iffull
\begin{proof}
If all the entries in the Counting Sort are unique, then the final for loop never overwrites a previous value, thus not incurring that cost. However, we still must pay the cost for the irreversible for loop, yielding an energy complexity of $\Theta(\lg n + \lg k)$.
\end{proof}
\fi

\begin{theorem} \label{thm:counting-sort-reversible}
Reversible Counting Sort can be done in $\Theta(n + k)$ time,  $\Theta(n + k)$ space, and  $0$ energy.
\end{theorem}

\iffull
\begin{proof}
Increments, addition, and simple for loops can all be done reversibly with constant-factor overheads. The assignment in the third loop is the only potential problem; however, we cannot have more than $n$ items overwritten, so if we log these values we suffer no more than a constant-factor overhead in space.
\end{proof}
\fi

\subsection{Graph Algorithms}
Frank \cite{Frank99} argues that Breadth-first Search and Depth-first search can be done reversibly. We reproduce this result in our model and give a different analysis.
\iffull
Breadth-first search requires keeping a list of visited nodes and a queue of nodes to be visited. The visited node is selected from the queue and then checked to see if it is the target node. If not, its neighbors are checked to see if they are in the list of visited nodes, and if not they are added to the queue. 
\fi

\begin{theorem}\label{thm:bfs}
Breadth-first Search runs in $\Theta(V+E)$ time,  $\Theta(V)$ space, and  $\Theta(wV+E)$ energy.
\end{theorem}
\iffull
\begin{proof}
In a BFS we visit every node at most once and we check to see if a node should be queued up no more than once for every edge. We also have this number of comparisons which each require one bit of energy if done irreversibly. Further, an irreversible queue requires the creation and destruction of pointers, and the list of visited nodes must be cleared, leading to $\Theta(V)$ words being overwritten.
\end{proof}
\fi

\begin{theorem}\label{thm:bfs-reversible}
Reversible Breadth-first Search can runs in $\Theta(V)$ time,  $\Theta(V+E)$ space, and  $0$ energy.
\end{theorem}
\iffull
\begin{proof}
Comparison and pointer following can be done reversibly. If the queue maintained as a doubly-linked list, then we can add and remove things reversibly via pointer swapping. Copying values into a list is also a reversible operation.
\end{proof}
\fi

\begin{corollary}\label{thm:dfs-reversible}
Reversible Depth First Search can runs in $\Theta(V+E)$ time,  $\Theta(V)$ space, and  $0$ energy.
\end{corollary}
\iffull
\begin{proof}
A DFS works exactly the same as a BFS except a stack of nodes to be visited is used instead of a queue. Stacks can be maintained reversibly with the same time, space, and energy complexity as queues, so the resulting DFS algorithm will have the same time, space and energy complexity as the BFS.
\end{proof}
\fi

\subsection{Bellman-Ford}
\iffull
\begin{pcode}
$\proc{Bellman-Ford}():$ \+
     $\keyw{for each}$ $v$ \keyw{in} vertices: \+
          \keyw{if} $v \neq$ source \keyw{then} \ul{v.distance} $ = \infty $  //rest are 0
          $\keyw{\ul{v.$predecessor$}} = null $ \-
		
     \keyw{for} $i = 1$ to vertices.size $- 1:$ \+
          \keyw{for each} edge $(u, v)$ with weight $w$ \keyw{in} edges: \+
               $(w, u.$distance$, \id{tempDist}) = (w, u.$distance$, u.$distance$ + w)$
               \keyw{if} $\id{tempDist} < v.$distance: \+
                    $v.$distance $= \id{tempDist}$
                    $(u, v.$predecessor$) = (u, u)$    //overwriting costs w energy \- \- \-
				
     \keyw{for each} edge $(u, v)$ with weight $w$ \keyw{in} edges: \+
          $(w, u.$distance$, \id{tempDist}) = (w, u.$distance$, u.$distance $ + w)$
          \keyw{if} $\id{tempDist} > v.$distance: \+
               \keyw{print} ``negative weight cycle detected''
               \keyw{fail}
\end{pcode}
\fi

\begin{theorem} \label{thm:bellman-ford}
Bellman-Ford runs in $\Theta(VE)$ time,  $\Theta(V+E)$ space, and  $\Theta(VEw)$ energy.
\end{theorem}
\iffull
\begin{proof}
In the primary nested for loops which iterate over the verticies and edges, we potentially overwrite the distance between two nodes with increasingly smaller distances, and overwrite the edges making up the shortest path. These each require $w$ energy to perform and are potentially called $\Theta(VE)$ times. This for loop is also what dominates the time complexity, as can be seen in any standard analysis. The algorithm requires storing the graph, as well as two arrays of size $V$, thus running in  $\Theta(V+E)$ space.
\end{proof}
\fi

\begin{theorem} \label{thm:bellman-ford-reversible}
Reversible Bellman-Ford runs in  $\Theta(VE)$ time,  $\Theta(VE)$ space, and  $0$ energy.
\end{theorem}
\iffull
\begin{proof}
First, at the end of each for loop involving the variable tempDist, we can subtract $distance[u]$ and $w$ from the value to ensure it is zero before we next need to set it, thus avoiding the energy cost of zeroing it out. Rewriting $v.$distance and $v.$predecessor do not appear to have an obvious way of making them reversible. We can; however, simply log the entire algorithm and pay the cost in writing their values to our log sheet whenever they are overwritten. This potentially requires $\Theta(VE)$ values to be stored in the log, resulting in a total space complexity of $\Theta(VE)$.
\end{proof}
\fi

\subsection{Floyd-Warshall}
\iffull
The Floyd-Warshall algorithm calculates the shortest paths between all pairs of nodes in $\Theta(V^3)$ time and $\Theta(V^2)$ space. Aside from initializing the path matrix to contain the edge weights, pseudocode for the algorithm is as follows \cite{clrs}:
\begin{pcode}
$\proc{FloydWarshall}()$: \+
     \id{for} $k=1$ \id{to} $n$:    \+   
            //simple for loop adds $\lg n$ extra energy or space
         \id{for} $i=1$ \id{to} $n$:   //another $\lg n$ \+
               \id{for} $j=1$ \id{}to $n:$    //one last $\lg n$ giving a total of $\lg^3n$ after nesting \+
                    $\id{path}[i][j] = \min(\id{path}[i][j], \id{path}[i][k]+\id{path}[k][j])$    
                            //unconstrained replacement takes w energy

\end{pcode}
\fi

Frank \cite{Frank99} argues that the Floyd-Warshall algorithm can be adapted to run reversibly with $\Theta(V^3)$ space. This is a substantial increase in space to make the program reversible and thus save energy.

\begin{theorem}\label{thm:floyd-warshall}
Floyd-Warshall runs in  $\Theta(V^3)$ time,  $\Theta(V^2)$ space, and  $\Theta(V^3w)$ energy.
\end{theorem}
\iffull
\begin{proof}
The algorithm acts on a matrix of size $v^2$ and performs $\Theta(V^3)$ updates which consist of a comparison, addition, and possible replacement. Each replacement requires $w$ energy, yielding a total energy complexity of $\Theta(V^3w)$.
\end{proof}
\fi

\begin{theorem} \label{thm:floyd-warshall-reversible}
Reversible Floyd-Warshall runs in $\Theta(V^3)$ time,  $O(V^3)$ space, and  $0$ energy.
\end{theorem}
\iffull
\begin{proof}
Simple for loops, as well as addition and minimum can all be performed reversibly; however, we will need to log the replacements done during the updates for the algorithm. Thus running the algorithm reversibly will require $O(V^3)$ space, despite constant-factor overheads in time.
\end{proof}
\fi

\subsection{All Pairs Shortest Path via $(\min, +)$ Matrix Multiplication}
Another algorithm for solving APSP involves using the adjacency matrix representation of a graph $A$ and noticing that the relaxations over the edges can be expressed by calculating a new matrix, $C$, whose entries are given by $c_{ij} = \underaccent{k}{\min} (a_{ik}+a_{kj})$. Further, this operation is associative, so we can speed up the calculation by using repeated squaring. Thus we have $O(\lg V)$ iterations over $(V^2)$ elements which take $O(V)$ time to compute. 
Frank \cite{Frank99} claims without proof that this leads to a reversible algorithm that runs in $\Theta(V^3\lg V)$ time and $\Theta(V^2\lg V)$ space. We give a proof of this result.

\begin{theorem} \label{thm:repeated-APSP-reversible}
Reversible APSP using repeated squaring with (min,+) matrix multiplication runs in $O(V^3\lg V)$ time,  $O(V^2 \lg V)$ space, and  $0$ energy.
\end{theorem}
\iffull
\begin{proof}
In this algorithm, information is lost when we overwrite our previous entries after calculating the minimum, an irreversible operation. To make this reversible, we can simply store a log of all intermediate matrices, using $O(V^2 \lg V)$ space. From each intermediate state, we can recompute the values needed to unroll the last step of the matrix squaring. We need only recompute these values a single time, leading to a constant-factor increase in time, giving the desired result. 
Let us take a closer look at an individual matrix multiply. Here we can loop over every new entry in our matrix and calculate each term individually. Given a graph represented by adjacency matrix $W = (w_{i,j})$, and a matrix $L^{(m)} = (l^{(m)}_{i,j})$ representing the shortest paths between two vertices with path length at most $m$. Each entry is updated as 
$$l^{(m+1)}_{i,j} = \underaccent{1\leq k \leq |V|}{\min} \left ( l^{(m)}_{i,j} + w_{k,j} \right )$$
This section of the algorithm perform $\Theta(V)$ reversible additions but also $O(V)$ irreversible operations to perform the $\min$ over all of the entries. We can log all of these operations with $O(wV)$ space overhead and then proceed to unwind leaving our new matrix entry. Now to unwind a whole matrix multiply, we will simply recalculate each individual entry, delete the copy, and unwind. Now that we know we can reverse each matrix multiply given the previous one, we will simply store all $O(\lg n)$ intermediate matrices which are calculated. This leads to $O(V^2 \lg V)$ space and a constant factor overhead in time.
\end{proof}
\fi

\begin{theorem} \label{thm:repeated-APSP}
APSP using repeated squaring with (min,+) matrix multiplication runs in $O(V^3\lg V)$ time,  $O(V^2)$ space, and  $O(wV^3\lg V)$ energy.
\end{theorem}
\iffull
\begin{proof}
The time analysis is the same as with the reversible case. In terms of space, we only need to store the current matrix, and thus only need $O(V^2)$ space. At every calculation of a new matrix element, we overwrite the previous one, expending a word of energy for each of these computations.
\end{proof}
\fi

We now present a new variation on APSP which demonstrates a non-trivial trade-off between energy and space. By exploiting reversible subroutines, we're able to reach the APSP with repeated squaring bounds on time and space, but beat it in energy cost. The reversible, semi-reversible, and standard APSP using repeated squaring demonstrate there are semi-reversible algorithms that actually achieve bounds not reached by the fully reversible or fully irreversible counterparts.

\begin{theorem}\label{thm:repeated-APSP-semi-reversible}
Semi-reversible APSP using repeated squaring with (min,+) matrix multiplication runs in $O(V^3\lg V)$ time,  $O(V^2)$ space, and  $O(wV^2\lg V)$ energy.
\end{theorem}
\begin{proof}
To begin, we will examine how each individual entry in the matrix is updated. Say we have a graph represented by adjacency matrix $W = (w_{i,j})$, and a matrix $L^{(m)} = (l^{(m)}_{i,j})$ representing the shortest paths between two vertices with path length at most $m$. Each entry is updated as 
$$l^{(m+1)}_{i,j} = \underaccent{1\leq k \leq |V|}{\min} \left ( l^{(m)}_{i,j} + w_{k,j} \right )$$
This subroutine runs in $O(V)$ time and $O(wV)$ energy and can thus be trivially made reversible by logging everything, using $O(V)$ time and space. We replace our normal update function with the new reversible one, and by Theorem~\ref{thm:fxncall} we have a new, more energy efficient algorithm. The subroutine does not use asymptotically more time than before, the temporary use of $O(V)$ space is much smaller than that needed to store the matrices and is freed upon completion of the subroutine, and the energy cost drops by a factor of $V$ which reduces the algorithms total energy cost by a factor of $V$.
\end{proof}

\section{Future Directions}
\label{sec:future}

This paper built up a framework for designing and analyzing the energy cost of
algorithms caused by irreversibility, and started the quest for positive
results for basic algorithms and data structures.  In many cases, we obtained
fully reversible versions of algorithms, but other problems seem more resistant.
For example, is there a reversible all-pairs shortest path algorithm with only
constant factor overheads in time and space?
We managed to give a reduced-energy semi-reversible algorithm for the problem,
but a fully reversible algorithm still seems elusive.
Shortest-path algorithms more generally seem like a category that are
difficult to make reversible, as they use very little space and make frequent
use of rewriting old values.  We anticipate other graph problems such as
max-flow/min-cut may also be challenging and interesting for similar reasons.

There are more fundamental algorithms that should be given high priority given
their use in many other results: hashing, predecessor data structures
(e.g., van Emde Boas trees), max-flow/min-cut, Fast Fourier Transforms,
and dynamic programming.
Geometric algorithms offer more nontrivial challenges to attain reversibility,
such as line intersection, orthogonal range finding, convex hull, and Delaunay
triangulations.
We also see the field of machine learning being an interesting target for
analysis in the semi-reversible model: these algorithms often have
significantly higher time complexities than space complexities, fundamental
updates (such as Bayes' rule) which appear reversible, and many conditional
updates or data overwrites.

One important question for any practical application is how to deal with
long-running programs. Although we are perfectly happy to log some auxiliary
information during the execution of a specific program, it may be more
problematic to maintain reversibility for the entire operating system of a
computer or a long-lived database. This is an area we believe ideas like
semi-reversibility and periodic rebuilding will become particularly important.

There are some areas where we see slight extensions of the model opening up
interesting questions. First, incorporating randomness seems a practical
necessity and carries interesting thermodynamic implications depending on how
it is modeled. Assuming there is an energy cost associated with the production
of randomness (say, equal to the number of random bits), this may give further
reason to investigate exactly how much randomness is needed for an algorithm's
correctness. Streaming algorithms and other models where the working space is
much smaller than the problem input seem like a rich source of problems.
Because we now use sublinear space, our trivial transform is no longer
applicable.  Further, the larger the gap in space and time, the less ability
we have to accrue garbage. Finally, succinct data structures, which try to
minimize the bits of space used up to sublinear factors, seem like another
challenge: many of our transforms double or triple the space being used by an
algorithm, while in the succinct setting, this overhead must be considered.

Finally, a major open direction is to obtain lower bounds.  The additional
constraints on semi-reversible algorithm design might allow showing algorithms
cannot be obtained without some minimum time-space-energy trade-off.

\section*{Acknowledgments}

We thank Martin Demaine and Kevin Kelley for helpful early discussions about
this project, in particular early formulations of the models of computation.
We also thank Maria L. Messick and Licheng Rao for useful discussion and help clarifying our pseudocode and model.

\let\realbibitem=\bibitem
\def\bibitem{\par \vspace{-1.5ex}\realbibitem}

\bibliographystyle{alpha}
\bibliography{EnergyBib}
\end{document}